\documentclass[journal]{IEEEtran}

\usepackage{color}
\usepackage{graphics}
\usepackage{graphicx}
\usepackage{epsfig}
\usepackage{amsmath,amssymb}
\usepackage{mathrsfs,amsfonts,fancyhdr}
\usepackage{multirow}
\usepackage{multicol}
\usepackage{cite}
\usepackage{amsthm}
\newtheorem{thm}{Theorem}[]
\newtheorem{cor}{Corollary}[]
\newtheorem{lem}{Lemma}[]

\usepackage{subfigure} 

\begin{document}

\title{{Link Regime and Power Savings of Decode-Forward Relaying in Fading Channels}}
\author{Lisa Pinals,~\IEEEmembership{Student Member,~IEEE,}
		and~Ahmad Abu Al Haija,~\IEEEmembership{Member,~IEEE}
        and~Mai Vu,~\IEEEmembership{Senior Member,~IEEE}
        \thanks{This work has been supported in part by the Office of Naval Research (ONR, Grant N00014-14-1-0645) and National Science Foundation Graduate Research Fellowship Program (NSF, Grant No. DGE-1325256). 
        
        L. Pinals and M. Vu are with the Department of Electrical and Computer Engineering, Tufts University, Medford, MA 02155 USA (e-mails: lisa.pinals@tufts.edu,   maivu@ece.tufts.edu).

A. Abu Al Haija is with the Department of Electrical and Computer
Engineering, University of Alberta, Edmonton, AB T6G 1H9, Canada (e-mail: abualhai@ualberta.ca).}}

\maketitle
\begin{abstract}
In this paper, we re-examine the relay channel under the decode-forward (DF) strategy. Contrary to the established belief that block Markov coding is always the rate-optimal DF strategy, under certain channel conditions (a link regime), independent signaling between the source and relay achieves the same transmission rate without requiring coherent channel phase information. Further, this independent signaling regime allows the relay to conserve power. As such, we design a composite DF relaying strategy that achieves the same rate as block Markov DF but with less required relay power. The finding is attractive from the link adaptation perspective to adapt relay coding and relay power according to the link state. We examine this link adaptation in fading under both perfect channel state information (CSI) and practical CSI, in which nodes have perfect receive and long-term transmit CSI, and derive the corresponding relay power savings in both cases. We also derive the outage probability of the composite relaying scheme which adapts the signaling to the link regime. Through simulation, we expose a novel trade-off for relay placement showing that the relay conserves the most power when closer to the destination but achieves the most rate gain when closer to the source.
\end{abstract}


\IEEEpeerreviewmaketitle
\section{Introduction}\label{sec:intro}
Relay aided communication has shown to be beneficial to wireless networks by improving throughput and coverage \cite{relayIntro2,d2d}. Given that wireless transmit power is often limited, it is of interest to design a relaying scheme that uses the minimum transmit power but still achieves the maximum possible rate. In this paper we reexamine the basic relay channel, in which a relay aids the transmission between the source and destination, with regard to relay power consumption. We assume a direct link between the source and destination, an appropriate model for wireless communication.

Minimizing relay transmit power allows the relay to operate longer and improves the overall network performance. With less required relay transmit power, the relay consumes less of its own resources to help another node, thus increasing the incentives for an idle node to utilize its own resources to relay the information of other nodes. Relays can also enable reductions of network energy consumption without complicated infrastructure modifications due to the reduced transmission distances\cite{greenmag1}. Further, by transmitting below full power, the relay creates less interference to other nodes in the network, resulting in better network performance.

\begin{figure}[t]
   \begin{center}
		\includegraphics[width=0.45\textwidth]{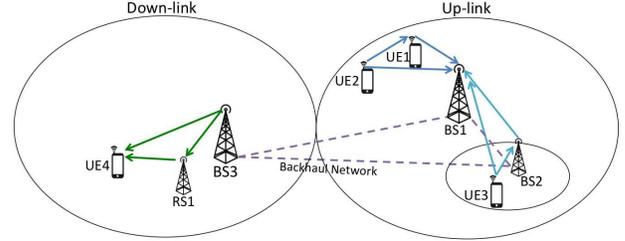}
    \caption{Examples of application of the relay channel in a cellular network.}
    \label{fig:motivation}
    \end{center}
\end{figure}
\subsection{Relay Channel Integration in Larger Networks}
This basic relay channel model could readily be applied as a component in a larger wireless network. In an ad hoc setting, relaying can be applied synonymously among all users \cite{adhoc1,adhoc2}. A third user can act as a relay, such as in emergency message dissemination in vehicular ad hoc networks. In future cellular systems, relaying can leverage device-to-device (D2D) communication to improve transmission rates and spectral efficiency \cite{d2d}.

In cellular networks, the use of relaying among user devices and base stations becomes more feasible due to network densification. For example, consider the network in Fig. \ref{fig:motivation} with two macrocells and a femtocell. This heterogeneous network contains several example applications of the relay channel enabled by the proximity of users. In addition to the regular use of relay stations such as RS$_1$ to facilitate the down-link communication between BS$_3$ and UE$_4$, several other more spontaneous relaying scenarios can be envisioned. An idle user, UE$_1$, acts as a relay for the up-link communication between UE$_2$ and BS$_1$ by utilizing a D2D link as in \cite{hussain1}. A macro user, UE$_3$, can relay through a nearby femto base station, BS$_2$, to improve its up-link throughput to BS$_1$. In these examples, the relayed communications can leverage both the relay link and the direct link to increase the transmission rate beyond what is achievable with the direct link alone.

\subsection{Background and Related Works}
Fundamental relay strategies include amplify-forward (AF) \cite{ref5in2}, compress-forward (CF), and decode-forward (DF) \cite{ray2,ref4in2}. AF is a low complexity strategy in which the noise along with the received signal is amplified and retransmitted. In CF, the received signal is quantized instead of decoded at the relay to alleviate the effects of noise. Here we focus on the DF relaying strategy \cite{ray2,ref4in2} in which the effects of noise are removed completely at the relay by decoding the message, either fully or partially, before re-encoding it to transmit to the destination. For the basic relay channel, the common DF technique is based on block Markov superposition coding. When considering a more complicated scenario involving more nodes, then independent coding at the relay is also useful in reducing inter-node interference \cite{achievable,xie2007network}. However, this idea has not been applied in the basic relay channel due to the belief that block Markov coding is the best DF technique for this channel.

The schemes in \cite{ref5in2, myJDF, DF10, myoutJ, ErkS, ahh4, outC2, DF3, PDFE, myoutC} employ direct transmission and either coherent or independent DF relaying depending on the decoding ability of the relay or the relative order between the source-to-relay and direct links. DF relaying can achieve a full diversity order of $2$ if the transmission switches between direct transmission and DF relaying based on the outage at the relay \cite{ref5in2} or the link order between the source-to-relay and direct links \cite{myJDF, outC2}. To date, there is no outage analysis for a DF scheme that alternates between direct transmission, independent coding, and coherent block Markov coding depending on the link state that achieves the same performance as that of direct transmission and block Markov DF relaying.

Use of independent coding impacts the power allocation at the source and relay differently than block Markov coding. Block Markov coding requires power splitting at the source between old and new messages. When relay transmit power is a design constraint, the energy efficiency of block Markov coding is shown to be greater than that of AF given that the relay is not very close to the destination \cite{rpower1}. In contrast, independent coding allows full source power to be devoted to the current message; further, independent coding is straightforward to implement as source-relay phase coherency is not required. Our results in the two-way relay channel (TWRC) demonstrate that independent coding helps to increase rate and reduce relay transmit power in certain link-state regimes \cite{pinals}. Here we want to analytically understand if independent coding and power savings also apply to the basic relay channel and obtain a closed form optimal power allocation, a result that was not possible for the TWRC. This relay power savings has been recognized in the basic relay channel in half-duplex mode when minimizing the source transmit power under a transmission rate constraint, leading to a power optimal scheme that may not consume full power at the relay \cite[Proposition 7]{Fanny1}. Although this work finds different cases of optimal power allocation in closed form, it did not identify the link conditions for these different cases and for relay power savings to be optimal. Here we analytically identify the conditions under which the relay conserves power but still achieves the maximum rate.

\subsection{Main Results and Contributions}
The main contribution of this work is the analysis of the link adaptation and relay power savings in a composite DF relaying scheme consisting of block Markov coding and independent coding. In order to reach this result, we develop a composite DF scheme that maximizes the rate for a given set of link states and identifies link-state regimes in which a particular transmission technique is optimal. Contrary to the established belief that block Markov coding is optimal whenever the source-to-relay link is stronger than the direct link \cite{hsce612}, this composite scheme achieves the same rate in all cases but also results in power savings at the relay.

Next, we investigate the implications of these link-state regimes in terms of link adaptation in fading. Although the composite DF scheme is derived for fixed links, we are interested in characterizing performance in a fading environment. We compute the outage probability in closed form and show that the scheme achieves a full diversity order of 2 by considering the asymptotic outage behavior at high SNR. 

Further, we consider which components of CSI are necessary at the source and relay to adapt to instantaneous fading. In the literature, different forms of CSI have been studied, including perfect CSI, long-term CSI, or partial CSI as some estimate in between \cite{CSIT1,CSIT2,reviewer2}. However, there has not been much study on different types of CSI in the full-duplex relay channel. As the relay receives on one link and transmits on another simultaneously, there are more CSI options such as long-term CSI on the transmit link with perfect CSI on the receive link. In this paper we examine the performance of this more viable, alternative form of CSI which we term \textit{practical CSI}. Practical CSI requires substantially less overhead than perfect CSI and still achieves significant gain over just long-term CSI or direct transmission.

By utilizing independent coding when the source-to-relay link is just marginally stronger than the direct link instead of block Markov coding, the relay does not need to use full power to achieve the highest possible rate. It is demonstrated analytically that the independent coding link regime has a significant probability of occurrence for most node-distance configurations in fading. The expected relay power savings with both perfect CSI and practical CSI assumptions are derived in closed form. Further, results show that even with only long-term CSI, the relay saves a significant portion of power with only marginal effect on the outage performance. Through simulation, we expose a novel trade-off between consumed relay power and rate gain for relay placement. Specifically, the relay conserves the most power when closer to the destination but achieves the most rate gain when closer to the source.

The rest of this paper is organized as follows: Section II describes the channel model and transmission scheme; Section III presents the link-state based optimized rate; Section IV examines link adaptation in fading; Section V investigates the outage probability for this composite scheme; Section VI examines the expected relay power savings under different CSI assumptions; Section VII presents numerical results; and Section VIII concludes the paper.


\begin{figure}[t]
   \begin{center}
    \includegraphics[width=0.25\textwidth]{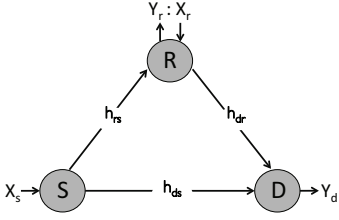} 
    \caption{Relay channel model.}\label{fig:Gausmod}  
    \end{center}
\end{figure}  

\section{Channel Model and Transmission Scheme}
\subsection{Channel Model}\label{sec:system_model}
In this paper, we consider a full-duplex one-way relay channel. Full duplex radio has emerged as a feasible wireless technique to improve spectral efficiency with significant new research demonstrating its feasibility \cite{fd2,fdmain}. Previously, the effects of self-interference could not be canceled at the hardware level and had to be dealt with at the signal processing level \cite{selfInterference1,selfInterference2}. However, researchers recently demonstrated a novel self-interference cancellation circuit and algorithm that provides the required 110dB of cancellation to reduce all self-interference to the noise floor in a Wi-Fi system, \textit{independent of the transmit power} \cite{fdmain}. To the best of our knowledge, no good model currently exists for this full-duplex residual self-interference. In the same way that MMSE error is independent from the estimated signal, we assume that the residual interference is independent of the transmit signal and thus can be modeled as another source of noise at the receive antennas. To incorporate the residual self-interference, it is equivalent to increase the noise power or decrease the transmit power. Further, results from a full-duplex scheme provide the basis for adaptation to half-duplex by later incorporating extra parameters of time or frequency slots required for half-duplex implementation.

Fig. \ref{fig:Gausmod} illustrates the one-way full-duplex relay channel model in which the source communicates with the destination with the help of a relay.
The discrete-time channel model is given as follows:
\begin{align}\label{recsig}
Y_{r}&=h_{rs}X_{s}+Z_{r}\nonumber\\
Y_{d}&=h_{ds}X_{s}+h_{dr}X_{r}+Z_{d},
\end{align}
\noindent where $Z_{r},Z_{d}\sim$ $\cal{CN}$(0,1), are independent complex Gaussian noises with zero mean and unit variance. $Y_{r}$ and $Y_{d}$ are the signals received by the relay and destination respectively and $X_{s}$ and $X_{r}$ are the signals transmitted from the source and relay respectively. The source (relay) input power is constrained to be at most $P_s$ ($P_r$).


The link gain coefficients are assumed to be affected by complex value Rayleigh fading as follows:
 \begin{align}\label{fadingch}
  h_k&=\tilde{h}_k/(d_k^{\gamma_k/2})=g_ke^{j\theta_k},\nonumber\\
   g_k&=|\tilde{h}_k|/d_k^{\gamma_k/2}, \quad k\in \{rs,ds,dr\}
 \end{align}
where $\tilde{h}_k$ is the small scale fading component distributed according to ${\cal CN}(0,1)$. The large scale fading component follows a pathloss model in which $d_k$ is the distance between two nodes and $\gamma_k$ is the attenuation factor. Let $g_k$ and $\theta_k$ be the amplitude and the phase of a link coefficient respectively, then $g_k$ follows a Rayleigh distribution while $\theta_k$ has uniform distribution between $[0,2\pi]$. We assume the distances among nodes are sufficiently large compared to the wavelength such that all links are independent.

The channel gain coefficients are typically complex value. However we assume the phases of these channel coefficients are known at the respective transmitters so that coherent transmission is possible, and the full channel coefficients are known at the corresponding receivers. The phase knowledge that allows coherent transmission is a standard assumption in coherent relaying literature \cite{csi,ray2,ahh4}. As such, the achievable rate depends only on the $amplitude$ of the link coefficients, denoted by $g_{*}$. Next, we describe the transmission scheme employed with this channel model.

We consider no interference in this model as we focus on the fundamental performance of composite DF relaying. When deployed in a larger setting such as in an ad hoc or cellular network, the received signals at the relay and destination will be perturbed by interference in addition to noise. While interference reduces the achievable rate and outage performance, we expect that composite DF relaying will keep outperforming the existing DF schemes in the presence of interference. The outage analyses in Section \ref{sec:OUTA} can be generalized to include interference in the noise terms by specifying new distributions of these terms that reflect both thermal noise and interference at the relay and destination. This topic is outside the scope of the current paper. The analysis in this paper, however, should provide the basis for further analysis with interference in future works.

\subsection{Transmission Scheme}
We propose a composite scheme that contains both coherent block Markov coding and independent network coding. These two distinct signal structures have different impacts on the transmit power and achievable rate. Here we introduce the composite scheme that contains both components and will analyze the necessity of each technique.

\subsubsection{Transmit signal design}
In each transmission block $i$, the source encodes new information $(m_i)$ by superposing it onto the old information $(m_{i-1})$ as in block Markov coding. The source transmits the signal that conveys this superposed codeword for $m_i$ and $m_{i-1}$. Next, the relay decodes $m_i$ and then performs both block Markov coding and independent coding. Specifically, the relay generates a new codeword for the decoded message $m_i$ independently from the source codeword and superposes this onto the source codeword for $m_i$. It transmits the signal of this superposed codeword (for $m_i$) in block $i+1$. The destination employs backward decoding: in block $i$, it decodes $m_{i-1}$ given that it knows $m_i$ from the decoding in block $i+1$.

Using Gaussian signaling, the source and relay construct their respective transmit signals $X_s$ and $X_r$ in block $i$ as follows:
 \noindent
\begin{align}\label{sigtr}
X_{s,i}&=\sqrt{\alpha_s}W_{s}(m_{i-1})+\sqrt{\beta_s}U_{s}(m_i),\\
X_{r,i}&=\sqrt{k_s\alpha_s}W_{s}(m_{i-1})+\sqrt{\beta_r}U_r(m_{i-1}),\nonumber
\end{align}
\noindent where $W_{s},U_{s}$ and $U_r$ are i.i.d Gaussian signals ${\cal N}(0,1)$. Note the superposition signal structure in which $U_{s}$ and $U_r$ are superposed on $W_s$ at the source and relay respectively. Here, $\beta_s$ and $\beta_r$ are the transmission powers allocated for signals $U_{s}$ and $U_{r}$, respectively; $\alpha_s$ and $k_s\alpha_s$ are the transmission powers allocated for signal $W_s$ by the source and relay, respectively. $k_s$ is a scaling factor that relates the power allocated to transmit the same message at the source and relay in the block Markov signal structure. These power allocation parameters satisfy the following power constraints:
\noindent
\begin{align}\label{powcsch2}
\alpha_s+\beta_s&\leq P_s,\;\;k_s\alpha_s+\beta_r\leq P_r.
\end{align}
where $P_s$ and $P_r$ are the maximum transmit powers of the source and relay respectively.

When $\alpha_s \neq 0$, the source sends not only the new messages in each block, but also repeats the message of the previous block. This retransmission due to block Markov coding creates a coherency between the signal transmitted from the source and the relay that ultimately results in a beamforming gain but requires power splitting at the source. In addition to block Markov coding, the relay also creates a new signal $U_r$ that independently encodes the message. This independent coding allows the source to devote full power to the new message of that block. Thus the two techniques have different implications on the power allocation at the source which ultimately affects the achievable rate depending on the link state.

\subsubsection{Decoding}
At the relay, decoding is simple and is similar to the single user case. The received signal in each block at the relay is
\noindent
\begin{align}\label{Yr}
Y_r\!&=g_{rs}(\sqrt{\alpha_s}W_s+\sqrt{\beta_s}U_s)+Z_r
\end{align}
\noindent
In block $i$, the relay already knows signal $W_s$ (which carries $m_{i-1}$), and is interested in decoding $U_s$ (which carries $m_{i}$), which it can perform after extracting known $W_s$.

Backward decoding is performed at the destination starting at the last block, in which the destination decodes just the previous block's information. The received signal in each block at the destination is
\noindent
\begin{align}\label{Y2}
Y_d\! &= \!g_{ds}(\sqrt{\alpha_s}W_s\!+\!\sqrt{\beta_s}U_s)\!+\!g_{dr}(\sqrt{k_s\alpha_s}W_s\!+\!\sqrt{\beta_r}U_r)\!+\!Z_d.
\end{align}
\noindent
Assuming that the destination has correctly decoded $m_{i}$, then in block $i$, the destination knows $U_{s}$ and can decode $m_{i-1}$ by joint decoding \cite[p. 391-393]{hsce612}. 

In this section, we use backward decoding at the destination to analyze the composite scheme but there are other alternative decoding techniques including forward sliding window decoding and sequential decoding. Both of these alternative destination decoding techniques achieve the same transmission rate as backward decoding for the basic relay channel considered, but have different implications on outage performance as well as decoding delay. Backward decoding makes reliability simplest to analyze as outage events are contained in a single transmission block, but is not as practical to implement because decoding can only begin after the last block is received, resulting in a large delay. Sliding window decoding reduces the delay to just one block by using the received signals in two consecutive blocks to decode, but complicates outage analysis if the channel changes over these two blocks.

Sequential decoding at the destination provides the lowest complexity option but requires some changes in the operation at the relay. Specifically, the relay performs random binning for the decoded message, $m_i$, to generate a codeword for its bin index and superposes this codeword onto the source codeword for $m_i$. Each bin represents a group of messages, allowing the code rate to be reduced for the bin index information encoded by signal $U_r$, which can then be decoded sequentially at the destination to identify the bin index before decoding the original source message. Similar to sliding window decoding, outage events at the destination in sequential decoding occur across two consecutive blocks, hence complicating the outage analysis. For these reasons, we choose to focus on backward decoding as a method for achieving the same transmission rate in the basic relay channel, while providing a lower bound on outage performance.

\subsubsection{Achievable Rate}
The above scheme results in the following achievable rate:
\begin{thm}\label{nathr1df}
Using the composite DF transmission scheme, the following rate is achievable for the Gaussian relay channel.
\begin{align}\label{regionDF}
R&\leq C\left(g_{rs}^2\beta_s\right)=J_1,\\
R&\leq C\left(\!g_{ds}^2(\alpha_s\!+\!\beta_s)\!+\!g_{dr}^2\left(\alpha_sk_s\!+\!\beta_r\right)\!+\!2g_{ds}g_{dr}\sqrt{k_s}\alpha_s\right)\!=\!J_2,\nonumber
\end{align}
\noindent with $C(x)\!=\!log(x\!+\!1),\ g_{*}$ as amplitudes of link coefficients, and power allocation factors satisfying (\ref{powcsch2}).
\end{thm}
\begin{proof}
Directly obtained from evaluating the mutual information rate expressions in \cite[p. 390-393, Theorem 16.2]{hsce612} with the transmit signal structures as in \eqref{sigtr}. $J_1$ results from decoding at the relay using the received signal in (\ref{Yr}) and $J_2$ results from decoding at the destination using the received signal in (\ref{Y2}).
\end{proof}


\section{Link-State Based Transmission Rate Optimization}\label{sec:analysis}
The composite scheme described in the previous section includes both coherent block Markov coding and independent coding at the relay. These two distinct signal structures have different impacts on the transmit power allocation and achievable rate. Here we analyze this general scheme to determine which parts of the signal are necessary for a given channel configuration. The goal of this section is to solve for the optimal power allocation for a given link-state. In doing so, we determine link-state regimes that describe which transmission technique is necessary for a given channel configuration. We also discuss the implications of the established link-state regimes and compare these new link-state regimes with the classical link-state regimes based solely on block Markov coding. Further, we numerically contrast the performance of this composite DF relaying scheme to other fundamental relaying strategies including compress-forward and amplify-forward relaying.

\subsection{Problem Formulation and Approach}\label{sec:approach}
We first formulate the source's rate optimization problem that integrates the power allocation factors for both the block Markov coding and independent coding techniques subject to the rate and power constraints as
\noindent
\begin{align}\label{optimization}
\max \ & R\nonumber\\
\text{subject to}\ & R\leq \min{[J_1,J_2]},\ \alpha_s\!+\!\beta_s\leq P_s,\ k_s\alpha_s\!+\!\beta_r\leq P_r,\nonumber\\
& \alpha_s,k_s,\beta_s,\beta_r\geq 0.
\end{align}

We apply the approach in \cite{myjournal} by first analyzing the optimization problem from the dual variable space, then examining the primal rate constraints. Specifically, we question if the optimal Lagrangian dual variables are zero or strictly positive and identify valid combination of these optimal dual variables. Analysis of these valid dual variable combinations leads to the establishment link-state regimes in the composite DF relaying scheme.

Optimization problem (\ref{optimization}) is convex and the optimal power allocation parameters can be fully deduced from the KKT conditions for a given set of link states. The Lagrangian is written as:
\begin{align}\label{LagDT}
\mathcal{L}\! \triangleq\! &-\! R\! -\!\lambda_1(J_1\!-\!R)\!-\! \lambda_2(J_2\!-\!R)\!-\!\lambda_3(P_s\!-\!\alpha_s\!-\!\beta_s)\nonumber\\
&-\!\lambda_4(P_r\!-\!k_s\alpha_s\!-\!\beta_r)\!-\! \lambda_{5}\alpha_s\!-\!\lambda_{6}\beta_r,
\end{align}
\noindent where $\lambda_i, i\in\{1,2,...,6\}$ are the dual variables associated with the corresponding rate and power constraints. Note that all dual variables are non-negative; that is $\lambda_i \geq 0 \ \forall i$.

From (\ref{regionDF}) the associated dual variables with each rate constraint ($\lambda_1, \lambda_2$) can be either strictly positive if the corresponding rate constraint is tight, or equal to zero if the constraint is loose. As such, there are four combinations of optimal dual variables. However, we can immediately eliminate $\lambda_1\!=\!\lambda_2\!=\!0$ because at least one rate constraint must be tight. As such, there are three potential combinations of rate constraints: 1) $\lambda_1>0$, $\lambda_2=0$; 2) $\lambda_1=0$, $\lambda_2>0$; and 3) $\lambda_1>0$, $\lambda_2>0$. However, case 2 implies that $R=J_2<J_1$ for all values of $\alpha_s$. When $\alpha_s=P_s$ though, $J_1=0$, which implies that $J_2<0$, which is impossible. Therefore, case 2 is invalid. In Appendix A, we analyze cases 1 and 3. Next we summarize the results obtained from this analysis.

\subsection{Link-State Regimes and Optimal Power Allocation}
First we introduce an important lemma concerning the power usage of the source and relay:

\begin{lem}\label{lemmaUsersFullPwr}
Given any set of link states, the source will always use full power with composite DF relaying. The relay will use full power if the source performs block Markov coding. If the source does not perform block Markov coding, then the relay need not use full power.
\end{lem}
\begin{proof}\label{proofUsersFullPwr}
See Appendix A.
\end{proof}
The proof to Lemma \ref{lemmaUsersFullPwr} follows from the idea that when the source performs block Markov coding, it is always beneficial to increase the power of the block Markov component at both the source and relay. Next, Appendix A contains further evaluation of the two valid cases above. This analysis determines the optimal transmission technique based on the link-state and the optimal power allocation within each link-state regime.

\begin{figure}[t]
\begin{center}
		\includegraphics[width=0.38\textwidth]{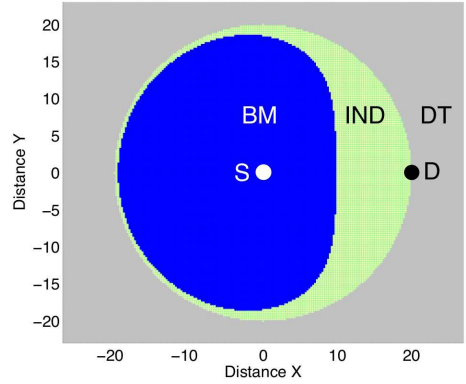}
    \caption{Optimal transmission as a function of distance in meters ($\gamma=3.6$). 'BM' indicates block Markov coding is optimal; 'IND' indicates that independent coding is optimal; 'DT' indicates that direct transmission is optimal.}\label{fig:simulation}
    \end{center}
\end{figure}

We now describe the analytical link-state regimes and the associated optimized DF transmission technique for each regime. For a particular set of channel gains, the associated link-state regime indicates which technique results in the best achievable rate.
\begin{thm}\label{tablethm}
In the Gaussian relay channel the following coding techniques are optimal in the defined link-state regimes:
\begin{align}\label{gr1}
&\mathcal{R}0: g_{rs}^2 \in[0,g_{ds}^2], &&\text{Direct  Transmission},\nonumber\\
&\mathcal{R}1: g_{rs}^2 \in(g_{ds}^2, g_{ds}^2+\frac{P_r}{P_s}g_{dr}^2] &&\text{Independent Coding},\nonumber\\
&\mathcal{R}2: g_{rs}^2 \in(g_{ds}^2+\frac{P_r}{P_s} g_{dr}^2,\infty), &&\text{Block Markov Coding}
\end{align}
\end{thm}
\begin{proof}
See Appendix A
\end{proof}

For the relay channel, we can find the optimal power allocation in closed form as shown in the following theorem.

\begin{thm}\label{thmoptimalpowerallocation}
The optimal power allocation within each link-state regime is as follows. Note that $\beta_s=P_s-\alpha_s$ since the source uses full power (Lemma \ref{lemmaUsersFullPwr}).
\begin{subequations}\label{optPwrAllocation}
\begin{align}\label{optPwrAllocationA}
&\mathcal{R}0:\ \alpha_s=k_s=\beta_r=0\\ \label{optPwrAllocationB}
&\mathcal{R}1:\ \alpha_s=k_s=0,\ \ \beta_r= \frac{g_{rs}^2-g_{ds}^2}{g_{dr}^2}P_s\\ \label{optPwrAllocationC}
&\mathcal{R}2:\ \beta_r=0,\ \  \alpha_s=\xi ,\ \  k_s=\frac{P_r}{\alpha_s}\\
&\text{where }\nonumber\\
&\xi\!\! =\!\! \left[\!\frac{-g_{ds}g_{dr}\sqrt{P_r}\!+\!\sqrt{g_{ds}^2g_{dr}^2P_r\!-\!g_{rs}^2\left(g_{dr}^2P_r\!+\!g_{ds}^2P_s\!-\!g_{rs}^2P_s\right)}}{g_{rs}^2}\!\right]^2\nonumber
\end{align}
\end{subequations}
\end{thm}
\begin{proof}
See Appendix A
\end{proof}

Fig. \ref{fig:simulation} provides an example of these link-state regimes on a 2D plane. The distance between the source and destination is fixed at 20 meters while the relay location varies over the entire plane. A path loss exponent of 3.6 without any small scale fading is assumed. For DF strategies, when the source-to-relay channel is strongest, coherent block Markov coding is optimal. When the relay is further away from the source, the source-to-relay channel is not quite as strong and independent coding is optimal. If the relay is very far from the source and the source-to-relay channel is weaker than the direct link, it is optimal to not utilize the relay.

\begin{figure}
\begin{center}
\includegraphics[width=0.3\textwidth]{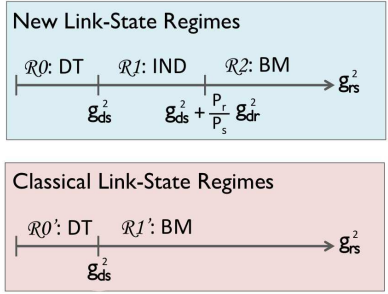}
    \caption{Link-state regimes: new versus classical}
     \label{fig:linkstateRegimes}
     \end{center}
\end{figure}

\subsection{Comparison of New and Classical Link-State Regimes}
Prior to the link-state regime analysis presented in this work, it was believed that block Markov coding was the optimal DF relaying strategy whenever the source-to-relay link was stronger than the direct link \cite{hsce612}. This belief results in the classical link-state regimes consisting of only two link-state regimes: direct transmission and block Markov coding. As demonstrated in this paper, a different set of link-state regimes contains three regimes with the addition of the independent coding regime between the direct transmission and block Markov link-state regimes as illustrated in Fig. \ref{fig:linkstateRegimes}. Even though the two sets of link-state regimes result in the same achievable rate, they have significantly different implications in terms of relay power consumption. Link-state regimes $\mathcal{R}0$ and $\mathcal{R}2$ are the same in the classical version as in the new version. The relay power savings are obtained in the new link-state regime $\mathcal{R}1$ through independent coding.

Power savings are achievable in link-state regime $\mathcal{R}1$ because in this link-state regime, the relay performs only independent coding (Theorem \ref{tablethm}). Without the block Markov signal component, $\alpha_s=k_s=0$ and the expended power at the relay is equivalent to $\beta_r= \frac{g_{rs}^2-g_{ds}^2}{g_{dr}^2}P_s$ from \eqref{optPwrAllocationB}. For the ranges of source-to-relay link gains within $\mathcal{R}1$, the required relay transmit power is always less than the maximum $P_r$, except at the upper boundary of the regime when $\beta_r=P_r$. Thus regime $\mathcal{R}1$ allows the relay to conserve power while still achieving the maximum possible rate. Even though this power savings applies only in one regime, the probability of link-state regime $\mathcal{R}1$ is non-negligible in fading as will be demonstrated in Section \ref{sec:OUTA}. This in turn leads to significant power savings in most fading channel environments.

In link-state regime $\mathcal{R}1$, the source-to-relay link magnitude is less than the combination of the source-to-destination and relay-to-destination link magnitudes as shown in Theorem \ref{tablethm}. As such, any information that the relay can decode from the source can also be decoded at the destination by combining the incoming links from the source and relay. Thus it is optimal for the source to utilize full power for the new message and not to repeat the message of the previous block, as in coherent block Markov coding. In this case, the relay needs to expend just enough power to send the new message over the relay-to-destination link. Consuming more power than $\beta_r$ in link-state regime $\mathcal{R}1$ is wasteful as it does not improve the rate.

\subsection{Comparison of Fundamental Relaying Strategies}
Next we contrast the achievable rate of composite DF relaying to other fundamental relaying strategies, including the traditional block-Markov based DF, amplify-forward (AF), compress-forward (CF), and the cut-set bound.  We do not include compute-forward in this comparison as compute-forward applies to a multi-source channel \cite{compute}. CF outperforms AF but is significantly more complex to implement \cite{hsce612}. The CF strategy is known to outperform traditional block-Markov DF when the relay is close to the destination, but is inferior to DF when the SNR of the source-to-relay link exceeds a certain threshold \cite{ref87}.

In Fig. \ref{fig:compareStrategies} we compare the achievable rates of AF, CF, block-Markov DF, composite DF relaying, and the cut-set bound. The simulation setup is the same as that of Fig. \ref{fig:simulation}, where the relay location varies along the line connecting the source and destination. The composite DF relaying scheme in this paper achieves the same rate as that of traditional DF relaying, but results in power savings at the relay in link-state regime $\mathcal{R}1$, indicated in light green. As established, when the relay is close to the source, DF outperforms CF; when the relay is close to the destination, CF begins to outperform DF. Note however, that there exists a range of relay locations in which the relay is closer to the destination and independent DF relaying still outperforms CF relaying, while saving significant relay power. In Fig. \ref{fig:compareStrategies}, this range applies from approximately 10 meters to 12.5 meters. This result highlights the novelty of the composite DF relaying scheme: the performance of independent DF with reduced relay power can match that of block-Markov DF and outperform that of CF, even when the relay is closer to the destination than to the source.

\begin{figure}[t]
   \begin{center}
    \includegraphics[width=0.42\textwidth]{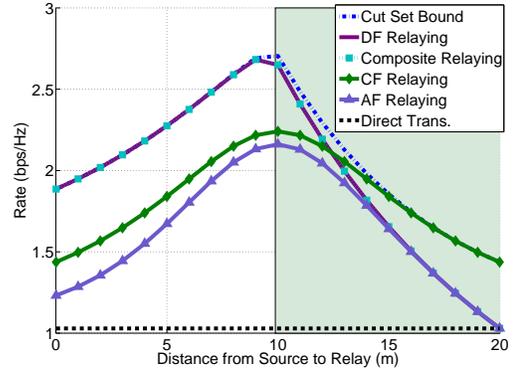} 
    \caption{Comparison of relaying strategies. The regions in which the composite DF relay performs block Markov coding (white) or independent coding (green) are marked on the figure. ($\gamma=3.6$, SNR$=5$dB)}\label{fig:compareStrategies}      
    \end{center}
\end{figure}  

\section{Link Adaptation in a Fading Environment}\label{sec:csi}
The link-state regimes are directly useful in a fading environment in order to perform link adaptation. We utilize the link regime results of the previous section in order to perform link adaptation of the composite scheme in a fading environment. Such adaptation requires various components of CSI at all nodes. Receiver CSI is a standard assumption and can be obtained directly by channel estimation. Transmitter CSI can be obtained via feedback from other nodes or via reciprocity. One particular node's knowledge of the link gains between the other two nodes would have to come from feedback. Here we discuss what CSI is necessary to implement the presented composite scheme. Specifically, we examine link adaptation based on perfect receive and long-term transmit CSI, a practical form of CSI, to reduce channel acquisition overhead.

\subsection{Required CSI to Adapt to Instantaneous Fading}
Certain components of CSI are required at the source and relay to compute both the appropriate link-state regime and the corresponding power allocation factors. Note that the link-state regimes in Theorem \ref{tablethm} can also be represented in terms of the received SNRs as
\begin{align}\label{rxSNR}
\mathcal{R}0: &\gamma_s \leq \gamma_o, \ \ \ \ \mathcal{R}1: \gamma_o<\gamma_s \leq \gamma_d,\ \ \mathcal{R}2: \gamma_s > \gamma_d\\
\text{where } &\gamma_s=g_{rs}^2P_s, \ \ \ \ \ \  \gamma_o=g_{ds}^2P_s, \ \ \ \ \  \gamma_d=\gamma_o+g_{dr}^2P_r.\nonumber
\end{align}
\noindent Here $\gamma_s$ is the instantaneous received SNR at the relay, $\gamma_o$ represents the received SNR at the destination if the source uses direct transmission, and $\gamma_d$ represents the received SNR at the destination when the relay is fully utilized in the independent mode.

To determine the appropriate link-state regime, the source and relay need information about the SNRs in \eqref{rxSNR}. Simply to determine whether the relay is utilized, the relative magnitude between the direct link and the source-to-relay link must be known by the source and relay. If the source-to-relay link is stronger than the direct link, then the relative magnitude among $\gamma_s$ and $\gamma_d$ must be known to discern between regimes $\mathcal{R}1$ and $\mathcal{R}2$. Thus to determine which link-state regime is applicable, the source and relay must know the relative order among $\gamma_o$, $\gamma_d$, and $\gamma_s$.

To determine the power allocation within a link-state regime, different CSI is necessary depending on the link-state regime. First we rewrite the optimal power allocation in \eqref{optPwrAllocation} in terms of the defined SNRs as:
\begin{subequations}
\begin{align}\label{optPwrAllocationRxSNR}
&\mathcal{R}1:\ \alpha_s=k_s=0,\ \ \beta_r= \frac{\gamma_s-\gamma_o}{\gamma_d-\gamma_o}P_r \\
&\mathcal{R}2:\ \beta_r=0, \ \ \ k_s=\frac{P_r}{\alpha_s}\\
& \ \ \ \ \ \ \ \alpha_s\!\!=\!\!\left[\!\frac{-\sqrt{\gamma_o \gamma_d \!-\! \gamma_o^2 }\!+\!\sqrt{\gamma_o(\gamma_d-\gamma_o)\! -\! \gamma_s(\gamma_d \!-\! \gamma_s)}}{\gamma_s}\!\right]^2\!P_s\nonumber
\end{align}
\end{subequations}
\noindent Once the link-state regime is accurately determined, the source needs no further CSI in link-state regimes $\mathcal{R}0$ and $\mathcal{R}1$ because $\beta_s=P_s$ in both regimes. In link-state regime $\mathcal{R}2$, the source needs CSI for the three SNRs, $\gamma_o$, $\gamma_d$, and $\gamma_s$, in \eqref{rxSNR} in order to compute $\alpha_s$.

For the relay to determine the power allocation, the SNRs in \eqref{rxSNR} must be known only in one link-state regime. In $\mathcal{R}0$ and $\mathcal{R}2$, the relay does not need any CSI because in $\mathcal{R}0$ the relay is not utilized and in $\mathcal{R}2$, the relay allocates full power to the block Markov component. However, in regime $\mathcal{R}1$, the relay needs knowledge of all three SNRs in \eqref{rxSNR} in order to compute $\beta_r$. We assume the relay always has perfect CSI for $\gamma_s$ due to receiver channel estimation, but in $\mathcal{R}1$ $\gamma_o$ and $\gamma_d$ must be fed back to the relay. Thus, each node only needs knowledge of the SNRs in one link-state regime for the composite scheme to adapt to instantaneous fading.

\subsection{Adaptation with Practical CSI and Perfect CSI}\label{sec:csiModels}
The perfect CSI model implies that each node has perfect instantaneous CSI of all links. Obtaining perfect instantaneous CSI in a fading environment is challenging and unrealistic in fast fading, in which the system may be unable to adjust to rapid channel variations. Here we also investigate other CSI models. The long-term CSI model implies that each node only knows the average gain of all links, including the incoming links, forward links, and other non-local links. We also consider a more realistic model called practical CSI, in which nodes have perfect receive CSI of incoming links, long-term transmit CSI of the forward links, and long-term CSI for all other non-local links. Even though this practical CSI appears intuitive and seems to be used frequently in the literature, it is used mostly in the context of a single transmitter single receiver link. The use of a relay which both receives at and transmits from the same node creates a finer distinction among these different CSI models: long-term, practical, and perfect CSI.

Under the practical CSI model, the relay knows the instantaneous receive SNR, $\gamma_s$, through receiver channel estimation of the incoming link, but only knows the average $\bar{\gamma_d }$ and $\bar{\gamma_o}$ on the outgoing link and the non-local link. The source only knows the average $\bar{\gamma_s}$, $\bar{\gamma_d }$, and $\bar{\gamma_o}$. The relay uses the accurate instantaneous information about $\gamma_s$ only in its power allocation but not in the link-state regime identification. With practical CSI, the link-state regime is identified based only on the long-term channel statistics even though the relay has more accurate information about $\gamma_s$. This is because the source and relay must agree in their choice of link-state regime. However, in link-state regime $\mathcal{R}1$, the relay does use the instantaneous receive SNR $\gamma_s$ in computing $\beta_r$. This power allocation ($\beta_r$) at the relay is the only difference between the practical CSI and long-term CSI models. With the perfect CSI model, the source and relay both have perfect instantaneous knowledge of $\gamma_o$, $\gamma_d$, and $\gamma_s$, and the link-state regime and power allocation are adapted to each fading channel realization.

Power allocation parameters in each link regime with long-term CSI ($\beta_{r,\text{lt}},\alpha_{s,\text{lt}}$), practical CSI ($\beta_{r,\text{prac}},\alpha_{s,\text{prac}}$), and perfect CSI ($\beta_{r,\text{perf}},\alpha_{s,\text{perf}}$) in \eqref{optPwrAllocation} can be rewritten as:
\begin{align}\label{optPwrAllocationRxSNR2}
\mathcal{R}1:&\beta_{r,\text{lt}}= \frac{\bar{\gamma_s}-\bar{\gamma_o}}{\bar{\gamma_d}-\bar{\gamma_o}}P_r, \ \ \beta_{r,\text{prac}}= \frac{\gamma_s-\bar{\gamma_o}}{\bar{\gamma_d}-\bar{\gamma_o}}P_r, \\
& \beta_{r,\text{perf}}= \frac{\gamma_s-\gamma_o}{\gamma_d-\gamma_o}P_r \nonumber\\
\mathcal{R}2\!:\ &\alpha_{s,\text{lt}}\!\!=\!\!\left[\frac{\!-\sqrt{\bar{\gamma_o} \bar{\gamma_d} \!-\! \bar{\gamma_o}^2 }\!+\!\sqrt{\bar{\gamma_o}(\bar{\gamma_d}-\bar{\gamma_o}) \!-\! \bar{\gamma_s}(\bar{\gamma_d} \!-\! \bar{\gamma_s})}}{\bar{\gamma_s}}\right]^2\!\! \!P_s\nonumber\\
&\alpha_{s,\text{prac}}\!=\!\alpha_{s,\text{lt}}\nonumber\\
&\alpha_{s,\text{perf}}\!\!=\!\!\left[\frac{\!-\sqrt{\gamma_o \gamma_d \!-\! \gamma_o^2 }\!+\!\sqrt{\gamma_o(\gamma_d-\gamma_o) \!-\! \gamma_s(\gamma_d \!-\! \gamma_s)}}{\gamma_s}\right]^2\!\!\!\! P_s.\nonumber
\end{align}

Performance of composite DF relaying varies with different CSI models. With the perfect CSI model, the composite DF scheme will always outperform direct transmission because the relay is only used when the relaying rate is higher than the direct transmission rate. The same cannot always be said with long-term or practical CSI, however. For a small fraction of node-distance configurations, there are some nuances in the rate equations. Define the direct transmission rate as $J_o\!=\!\log(1+g_{ds}^2P_s)$. Due to the lack of precise CSI knowledge causing a sub-optimal power allocation in a few particular node-distance configurations, even though $\mathbb{E}\left[J_o \right]<\mathbb{E}\left[J_1 \right]$ always holds, $\mathbb{E}\left[R=\min\left(J_1,J_2\right) \right]$ is not necessarily greater than $\mathbb{E}\left[J_o \right]$ as $J_2$ could be less than $J_1$, causing the average relaying rate to fall below the average direct transmission rate. We note that these events are rare and only happen in a select few node-distance configurations, as illustrated in the numerical results Section \ref{sec:numresults}. Even though there is a slight loss in transmission rate, the affected regions and the rate loss are insignificant such that the advantages of such simpler and more realistic models as practical and long-term CSI prevail.

\section{Outage Analysis}\label{sec:OUTA}
The outage probability is another important performance measure for wireless services that require a sustained minimum target rate. In this section, we first derive the probability of each link-state regime, which helps in the subsequent derivation of the outage probability. We then derive the outage probability for the composite DF scheme in closed form, taking into account outages at both the relay and destination; we also examine the diversity order achievable by this composite scheme.

\subsection{Probability of Each Link-State Regime}\label{sec:prRegime}
We first analytically derive the probability of each link-state regime for a given distance configuration among the source, destination, and relay. From (\ref{fadingch}), note that $g_{k}^2=|\tilde{h}|^2/d_{k}^\gamma$ where $\tilde{h}$ is a Rayleigh fading component. Thus the squared amplitude of each link coefficient is an exponential random variable denoted as $g_{k}^2 \sim \text{exp}\left(\lambda_{k}\right)$. Note that for $g_{dr}^2$, we actually look at the distribution of $g_{dr}^2\frac{P_r}{P_s}$ as this is the term that appears in the link-state regime boundaries in Theorem \ref{tablethm}. Define the exponential parameters:
\begin{align}\label{expParameters}
&\lambda_{rs}=d_{rs}^{\gamma},\ \ \lambda_{ds}=d_{ds}^{\gamma},\ \ \lambda_{dr}=d_{dr}^{\gamma} \ \ \tilde{\lambda}_{dr}=\frac{P_s}{P_r}d_{dr}^{\gamma}.
\end{align}

\begin{lem}\label{theoremprLS}
For parameters in (\ref{expParameters}), the probability of each link-state regime (Theorem \ref{tablethm}) in Rayleigh fading is as follows:
\begin{align}\label{prLinkStates}
&\mathsf{Pr} (\mathcal{R}0)\!=\!\frac{\lambda_{rs}}{\lambda_{ds}\!+\!\lambda_{rs}}, \
\mathsf{Pr} (\mathcal{R}1)\!=\!\frac{\lambda_{rs}\lambda_{ds}}{\left(\lambda_{rs}\!+\!\lambda_{ds}\right)\left(\lambda_{rs}\!+\!\tilde{\lambda}_{dr}\right)},\nonumber\\
&\mathsf{Pr} (\mathcal{R}2)\!=\!\frac{\lambda_{ds}\tilde{\lambda}_{dr}}{\left(\lambda_{rs}\!+\!\lambda_{ds}\right)\left(\lambda_{rs}\!+\!\tilde{\lambda}_{dr}\right)}.
\end{align}
\end{lem}
\begin{proof}
See Appendix C.
\end{proof}

By employing Lemma \ref{theoremprLS} and simply knowing the distances between nodes, path loss exponent, and node transmit power constraints in a Rayleigh fading environment, the probability of each link-state regime can be calculated analytically. This closed form probability of each link-state regime will help in computing the outage in the next section and in computing the expected relay power savings in Section \ref{sec:fadingSection1}.

\subsection{Outage Probability}
Since we use backward decoding, outage events are confined to one transmission block and do not span multiple blocks, even though information is sent over two consecutive blocks with block Markov coding. In order to derive the overall outage probability, we analyze the outage events in each link-state regime. For link-state regime $\mathcal{R}0$, the outage probability is obtained as in  point-to-point communication under the condition that $g_{rs}\leq g_{ds}$. For link-state regime $\mathcal{R}1$, outage occurs only when the relay is in outage, since if the relay can decode the source's information, then so can the destination because of stronger combined links from the source and relay. For link-state regime $\mathcal{R}2$, outage can occur separately at the relay or destination. In this case, an outage can occur at the destination even when there is no outage at the relay. Next we formulate the outage probability.

\begin{lem}\label{outage formulation}
For a target rate $R$, given the relay rate constraints in Theorem \ref{nathr1df} and the direct transmission rate defined as $J_o=\log(1+g_{ds}^2P_s)$, the average outage probability $(\bar{\mathcal{P}}_{out})$ of the composite DF scheme can be formulated as follows:
\begin{align}\label{outtc}
\bar{\mathcal{P}}_{out}&=\mathcal{P}_{dt}+\mathcal{P}_{relay}+\mathcal{P}_{dest},\\
\;\text{with\ \ }\mathcal{P}_{dt}&=\mathsf{Pr}\left[\left( R> J_o\right) \cap \mathcal{R}0\right],\nonumber\\
\mathcal{P}_{relay}&=\mathsf{Pr}\left[ \left(R> J_1\right) \cap \left( \mathcal{R}1 \cup \mathcal{R}2\right) \right],\nonumber\\
\mathcal{P}_{dest}&=\mathsf{Pr}\left[ \left(J_2<R\leq J_1\right) \cap \mathcal{R}2\right]\nonumber
\end{align}
\noindent where $\mathcal{P}_{dt}$ is the outage probability of direct transmission that occurs link-state regime $\mathcal{R}0$; $\mathcal{P}_{relay}$ is the outage probability at the relay applicable in link-state regimes $\mathcal{R}1$ and $\mathcal{R}2$; and $\mathcal{P}_{dest}$ is the outage probability at the destination applicable only in link-state regime $\mathcal{R}2$.
\end{lem}
\begin{proof}
See Appendix B.
\end{proof}

Using the Rayleigh distribution of $g_{ds},$ $g_{rs}$ and $g_{dr}$, $\mathcal{P}_{dt},$ $\mathcal{P}_{relay}$ and $\mathcal{P}_{dest}$ are evaluated in closed form as in the following theorem.
\begin{thm}\label{cor_out}
The average outage probability $(\bar{\mathcal{P}}_{out})$ of the composite DF relaying scheme for a given power allocation can be evaluated as follows:
\begin{align}\label{outevl}
\bar{\mathcal{P}}_{out}&=\mathcal{P}_{dt}+\mathcal{P}_{relay}+\mathcal{P}_{dest},\;\text{with}\\
\mathcal{P}_{dt}&=1-e^{-\lambda_{ds}\beta_1^2}-c_1\left(1-e^{-\beta_1^2(\lambda_{rs}+\lambda_{ds})}\right)
\nonumber\\
\mathcal{P}_{relay}&=1-e^{-\lambda_{rs}\eta_1^2}-c_2\left(1-e^{-\eta_1^2(\lambda_{rs}+\lambda_{ds})}\right)\nonumber\\
\mathcal{P}_{dest}&=e^{-\lambda_{rs}\eta_1^2}\left(1-e^{-\lambda_{ds}\beta_1^2}\right)\nonumber\\
&\ -e^{-\lambda_{rs}\eta_1^2}\int_{0}^{\beta_1}2\lambda_{ds}g_{ds}
e^{-\left(\lambda_{ds}g_{ds}^2+\tilde{\lambda}_{dr}\zeta_1^2\right)}dg_{ds}\nonumber\\
\text{where } c_1&=\frac{\lambda_{ds}}{\lambda_{rs}+\lambda_{ds}},\quad c_2=\frac{\lambda_{rs}}{\lambda_{rs}+\lambda_{ds}}\nonumber\\
\beta_1&=\sqrt{\frac{2^R-1}{P_s}},\quad
\eta_1=\sqrt{\frac{2^R-1}{\beta_s}},\nonumber\\
\zeta_1&=\frac{-g_{ds}\sqrt{\alpha_s}+\sqrt{g_{ds}^2(\alpha_s-P_s)+2^R-1}}{\sqrt{P_s}}.\nonumber
\end{align}

\end{thm}
\begin{proof}
Using the Rayleigh distribution of $g_{ds},$ $g_{rs}$ and $g_{dr}$ as shown in Appendix B.
\end{proof}
Theorem \ref{cor_out} specifies the outage probabilities for a fixed power allocation. As shown in the following sections, these outage expressions for the composite DF scheme are instrumental in deriving the diversity order and in studying the impact on outage performance when conserving relay power in the independent coding link-state regime. Moreover, the results of Theorem \ref{cor_out} provide the basis for future works to optimize the power allocation for minimum outage probability, and to analyze system-wide outage performance when integrating the relay channel into a larger network.
\subsection{Diversity Analysis}\label{sec:DA}
The diversity order of a transmission scheme shows how fast the error rate decreases at high SNR and corresponds to the slope of the outage probability versus SNR. For the composite DF relaying scheme, the diversity order can be obtained by analyzing the outage probabilities in (\ref{outevl}) at high SNR. However, directly from the formulation in (\ref{outtc}), we can intuitively see that the outage probability for each link-state regime is proportional to $1/\text{SNR}^2$ as each regime requires at least two different links to be weak in order to have an outage at the relay or the destination. For $\mathcal{R}0$, outage occurs if $g_{ds}$ is weak and $g_{ds}\geq g_{rs}$. Hence, both $g_{ds}$ and $g_{rs}$ must be weak in order to have an outage in $\mathcal{R}0$. Similarly for $\mathcal{R}1$, an outage occurs when both $g_{ds}$ and $g_{rs}$ are weak. For $\mathcal{R}2$, an outage at the relay occurs if
$g_{rs}$ is weak and $g_{rs}\geq \sqrt{g_{ds}^2+\frac{P_r}{P_s}g_{dr}^2}$. Hence, all links ($g_{ds},$ $g_{rs}$ and $g_{dr}$) must be weak in order to have an outage at the relay, implying that relay outage probability in $\mathcal{R}2$ is proportional to $1/\text{SNR}^3$. If no outage occurs at the relay in $\mathcal{R}2$, then an outage at the destination occurs if
$g_{ds}$ and $g_{dr}$ are weak while $g_{rs}$ is strong, which leads to $\mathcal{P}_{dest}$ proportional to $1/\text{SNR}^2$.
The following corollary shows the asymptotic outage probabilities at high SNR.

\begin{cor}\label{thmDA}
The asymptotic outage probability of the composite DF relaying scheme at high SNR approaches the following values:
\begin{align}\label{asyb}
\mathbb{P}_{dt}&=\frac{\lambda_{ds}\lambda_{rs}(2^R-1)^2}{2P^2},\quad
\mathbb{P}_{relay}=\frac{\lambda_{ds}\lambda_{rs}(2^R-1)^2}{2bP^2},\\
\mathbb{P}_{dest}&=\frac{\lambda_{ds}\lambda_{dr}(2^R-1)^2}{2P^2}
\left(\sqrt{\frac{a}{a-1}}\sinh^{-1}\left(\sqrt{a-1}\right)-1\right),\nonumber
\end{align}
where $P=P_r=P_s$, $a=\alpha_s/P_s$, $b=\beta_s/P_s$ and $\sinh^{-1}$ is the inverse of the hyperbolic sine $\sinh$ function.
These expressions show that the diversity order is $2$, the maximum diversity for the basic relay channel.
\end{cor}
\begin{proof}
$\mathbb{P}_{dt}$ and $\mathbb{P}_{relay}$ in (\ref{asyb}) are obtained directly using the second order Taylor series expansion for the exponential functions in (\ref{outevl}). $\mathbb{P}_{dest}$ is obtained by taking the second order Taylor series expansion for the exponential functions in (\ref{outevl}) and evaluating the integral as shown in \cite[pp. 64,68]{tabofint}.
\end{proof}

\section{Relay Power Savings}\label{sec:fadingSection1}
In this section, we analyze the relay power savings in Rayleigh fading channels. In a fading environment, the power saving link-state regime, $\mathcal{R}1$, can occur with some probability for any node distance configuration. With the closed form probability of each link-state regime as derived in Section \ref{sec:prRegime}, we analyze power savings at the relay given different CSI assumptions. Specifically, we consider perfect CSI and practical CSI scenarios and analytically derive the average relay power savings in each case.

\subsection{Required Relay Power}
As shown in Section \ref{sec:analysis}, in link-state regime $\mathcal{R}1$, the relay need not use full power in order for the source to achieve the maximum rate. Specifically the relay power savings in $\mathcal{R}1$ with perfect CSI and practical CSI respectively are
\begin{align}\label{requiredRpowerR1}
S_{\text{perf}}(\mathcal{R}1) = P_r - \beta_{r,\text{perf}},  \quad S_{\text{prac}}(\mathcal{R}1) = P_r - \beta_{r,\text{prac}},
\end{align} 
\noindent with $\beta_{r,\text{perf}}$ and $\beta_{r,\text{prac}}$ as defined in \eqref{optPwrAllocationRxSNR2}. Even though relay power savings are identified only for link-state regime $\mathcal{R}1$, this regime occurs with non-negligible probability in a fading channel as demonstrated in Lemma \ref{theoremprLS}. This power savings is a novel result and fits with the projected theme of green communications in future wireless networks \cite{fullduplex}. 

For each distance configuration in a Rayleigh fading environment, there is a probability for the link-state to be in regime $\mathcal{R}1$ where the relay need not use full power. Link-state regime $\mathcal{R}1$ is significantly probable when the long-term channel statistics correspond to regime $\mathcal{R}1$ as in Fig. \ref{fig:simulation}. This provides ample opportunity for power savings at the relay. For the locations where the relay does not save much power but is still between the source and destination, especially when the relay is closer to the source, there is a more significant rate gain instead of relay power savings. Hence there exists a trade-off between relay power savings and rate gain, which will be further explored numerically in Section \ref{sec:numresults}.

\subsection{Average Relay Power Savings in a Fading Environment}
Next we compute the expected value of the relay power savings for a given distance configuration under the two different CSI assumptions discussed in Section \ref{sec:csi}. Specifically, the expected value of the relay power savings with perfect CSI and practical CSI are presented in closed form in the following theorem.

\begin{thm}\label{thmExpectedPwrSavings}
The expected relay power savings of composite DF relaying when both the source and relay have perfect CSI of $\gamma_o,\gamma_s,\gamma_d$ are
\begin{align}\label{pwrSavings2}
\mathbb{E}[\mathcal{S}_{\text{perf}}]&=\!\frac{\lambda_{ds}\lambda_{dr}\left[P_s\left[\ln{(\lambda_{dr})}\!-\!\ln{(\frac{P_r}{P_s}\lambda_{rs}\!+\!\lambda_{dr})}\right] \!+\!P_r\frac{\lambda_{rs}}{\lambda_{dr}}\right]}{\lambda_{rs}\left(\lambda_{rs}+\lambda_{ds}\right)}.
\end{align}
With practical CSI (perfect receive and long-term transmit CSI), the expected relay power savings are
\begin{align}\label{pwrSavingsltCSI_2}
\mathbb{E}[\mathcal{S}_{\text{prac}}] =
\left\{
	\begin{array}{ll}
	P_r - \frac{\lambda_{ds}\lambda_{dr}-\lambda_{rs}\lambda_{dr}}{\lambda_{rs}\lambda_{ds}}P_s  & \mbox{if } \bar{\gamma}_o < \bar{\gamma}_s \leq \bar{\gamma}_d \\
		0 & \mbox{else } 
	\end{array}
\right.
\end{align}
\end{thm}
\begin{proof}
See Appendix C.
\end{proof}
\noindent With perfect CSI, these relay power savings can be realized whenever the instantaneous link states satisfy the link conditions for $\mathcal{R}1$. Conversely, with practical CSI, the link-state regime depends only on the average link amplitudes because the source and relay must agree in their choice of link-state regime. Therefore, with practical CSI, the relay saves power only when the average link states satisfy the condition for link-state regime $\mathcal{R}1$ (i.e. if $\bar{\gamma}_o < \bar{\gamma}_s \leq \bar{\gamma}_d$). Assuming either perfect or practical CSI, with only knowledge of the distance configuration and path loss exponents, the mean relay power savings in a Rayleigh fading environment can be computed in closed form.

The relationship between the two expected power savings under perfect CSI and practical CSI in Theorem \ref{thmExpectedPwrSavings} is not explicit. The expected power savings with practical CSI can be larger than that of perfect CSI because with practical CSI, there is a trade-off between power savings and rate. Specifically, with practical CSI, the achievable rate is generally smaller than that with perfect CSI. For example, consider a node-distance configuration in which link-state regime $\mathcal{R}1$ is appropriate using practical CSI. In this case, the relay always conserves power as the link is assumed to be in link-state regime $\mathcal{R}1$. However, the rate in the practical CSI case would not be as high as that of perfect CSI because the power allocation is not adapted to each particular fade. Therefore, in this case, rate is sacrificed for relay power savings. This trade-off between rate and power savings using practical CSI will be numerically illustrated in the next section.

\begin{figure*}[t]
   \begin{center}
        \subfigure[]{
		\includegraphics[width=0.45\textwidth]{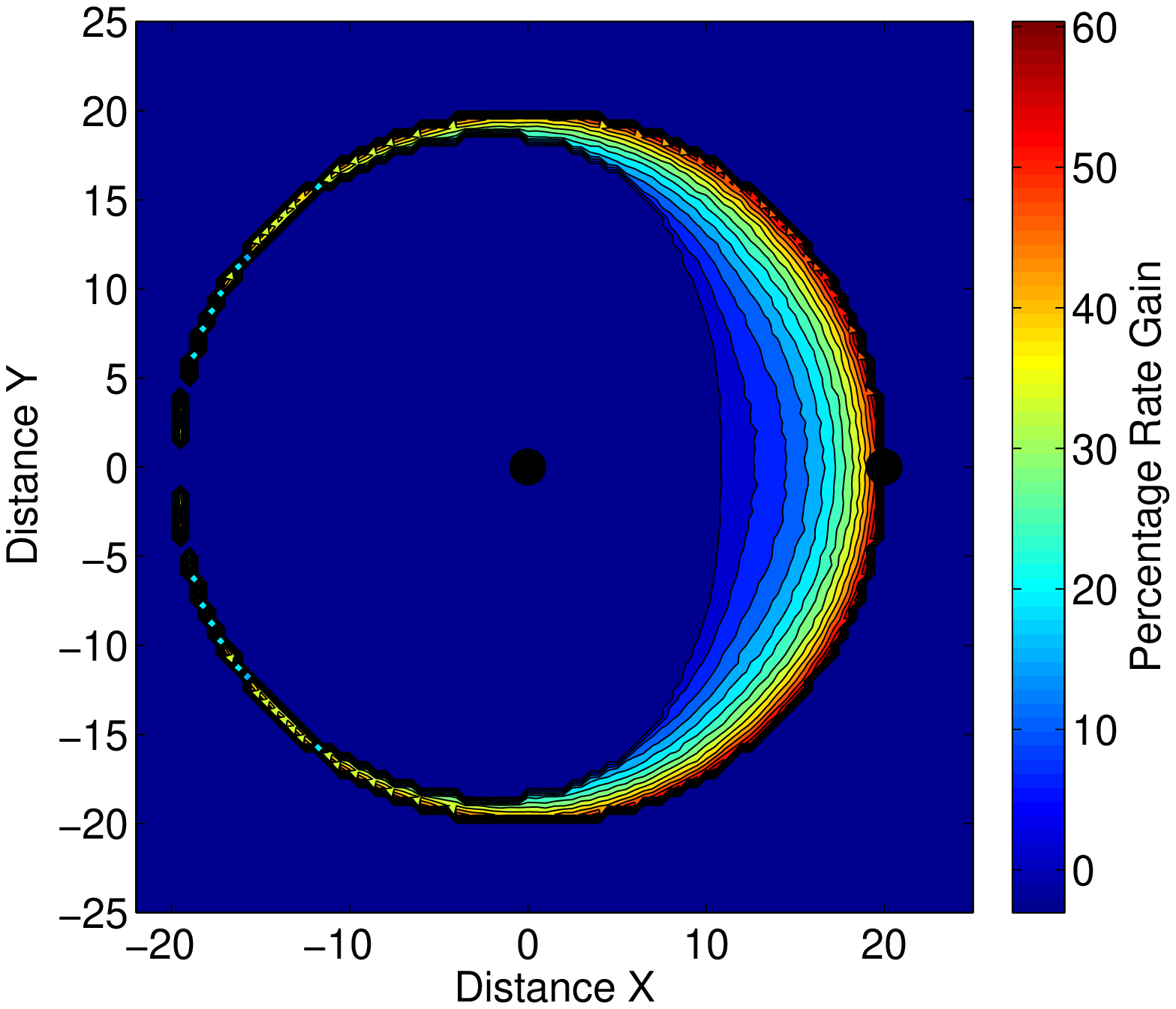}
		\label{fig:rateGain1}
    }
        \subfigure[]{
				\includegraphics[width=0.45\textwidth]{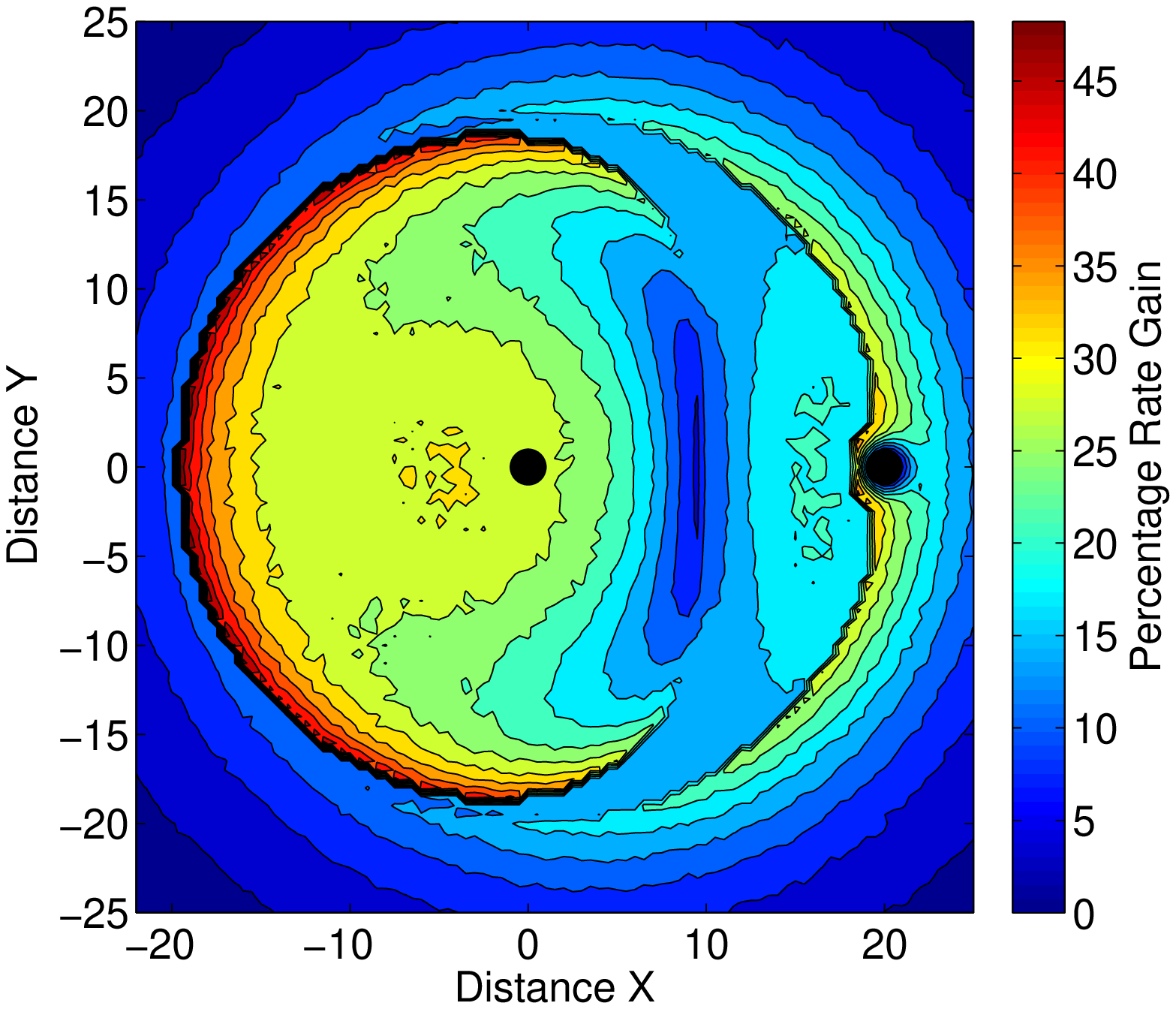}
				\label{fig:rateGain2}
    }
            \subfigure[]{
    \includegraphics[width=0.465\textwidth]{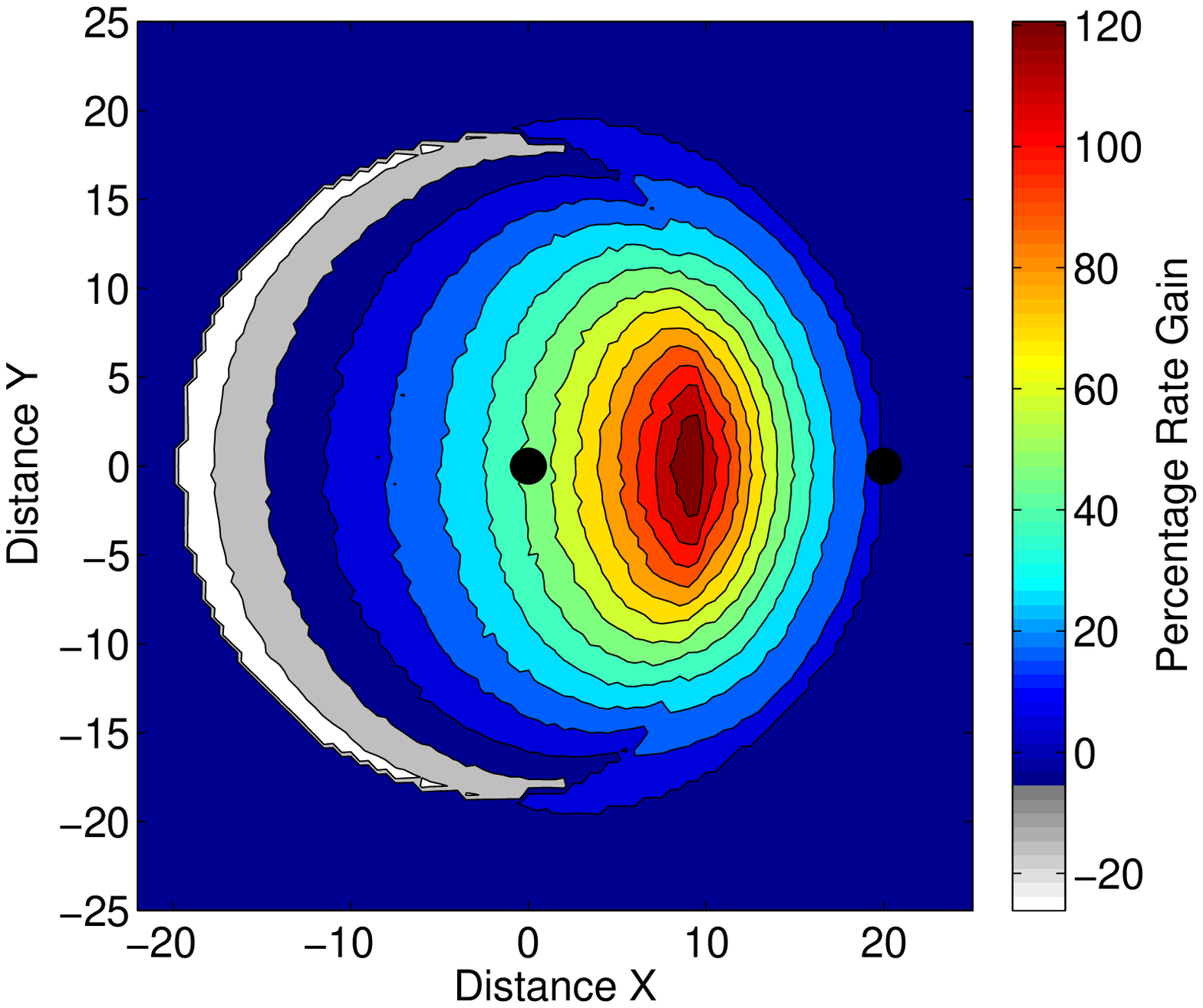}
				\label{fig:rateGain3}
    }
            \subfigure[]{
    \includegraphics[width=0.465\textwidth]{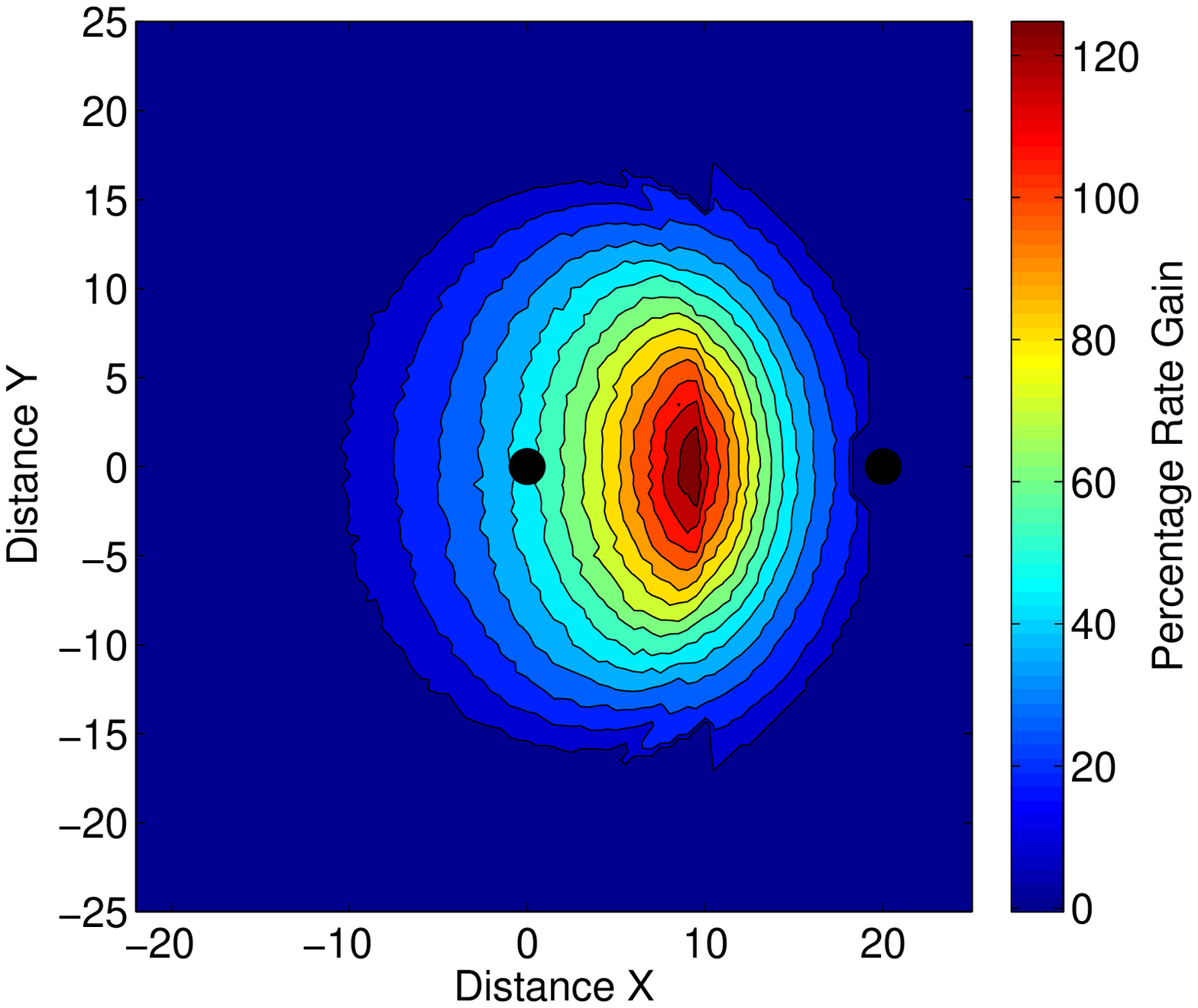}
				\label{fig:rateGain4}
    }
    \caption{Percentage rate gain of composite DF relaying with various models of CSI: a) practical CSI over long-term CSI; b) perfect CSI over practical CSI; c) practical CSI over direct transmission; d) practical CSI over direct transmission with ellipsoidal numerical threshold ($\gamma=3.6$, SNR$=5$dB)}
        \label{fig:rateGain}
    \end{center}
\end{figure*}

\section{Numerical Results}\label{sec:numresults}
In this section, we present numerical results to verify the analysis in Sections \ref{sec:csi}, \ref{sec:OUTA}, and \ref{sec:fadingSection1}. Define the average received $\text{SNR}$ at the destination for the signal from the source as follows:
\begin{align}\label{snreq}
\text{SNR}&=10\log\left(P/(d_{ds}^{\gamma})\right).
\end{align}
\noindent We set both the source and relay power equal to $P$ and pathloss equivalent to 3.6 where applicable. In spatial simulations, the source is fixed at [0,0] and destination at [20,0] (the same simulation setup as for Fig. \ref{fig:simulation}) and the relay location is varied over the whole plane using resolution of $1/2$ meter. At each relay location, 10,000 Rayleigh fading channels are generated using the inter-node distances for that particular relay location.

\subsection{Performance Comparison for Various CSI Models}
First we apply the link-state regimes to perform link adaptation. In Fig. \ref{fig:rateGain}, we compare the rates achievable with the various models of CSI (long-term CSI, practical CSI, and perfect CSI) as discussed in Section \ref{sec:csi}. Power allocation parameters are computed as in \eqref{optPwrAllocationRxSNR2} and the rate under each CSI model is averaged over all channel fading realizations.

Fig. \ref{fig:rateGain} demonstrates the usefulness of practical CSI. Recall from Section \ref{sec:csi} the difference between practical CSI and long-term CSI is that with practical CSI, the relay utilizes a more up-to-date value of $\beta_r$ due to instantaneous receive SNR $\gamma_s$. Compared to long-term CSI in Fig. \ref{fig:rateGain1}, practical CSI offers significant rate gain in the independent coding link-state regime when the relay is closer to the destination. Fig. \ref{fig:rateGain2} depicts the rate gain of perfect CSI over practical CSI. These rate gains are not substantial when the relay is between the source and destination, demonstrating that practical CSI is a viable alternative to perfect CSI in a fading environment.

In Fig. \ref{fig:rateGain3}, it is evident that using practical CSI with the composite scheme substantially outperforms direct transmission when the relay is between the source and destination. In these relay locations, employing the composite scheme with practical CSI results in up to 120\% rate gain over direct transmission, a significant gain. In some relay locations when the relay is farther away from the destination than from the source, however, composite DF is outperformed by direct transmission as discussed in Section \ref{sec:csiModels}; these areas are in grey and white in Fig. \ref{fig:rateGain3}.

The probabilities of the events discussed in  Section \ref{sec:csiModels} that cause direct transmission to outperform composite DF could be computed analytically, which would provide an exact rule for when to use the relay. However, because the affected region and the rate loss are both relatively small, the impact of an exact analysis would be marginal. Instead we apply a simplified numerical rule. Based on the results in Fig. \ref{fig:rateGain3}, the relay should be utilized if it falls within an ellipse with focal points (2,8) and (2,-8), or if the relay-to-destination distance is less than the source-to-relay distance. This geometry can be found numerically for each given source-to-destination distance and pathloss combination and accessed via a lookup table.

\subsection{Outage Performance}
We next provide numerical results for the outage probabilities of the considered composite DF scheme. In these simulations, we set the target rate $R(\text{target})=5\text{bps/Hz}$. We assume all links are Rayleigh fading channels such that the average channel gain for each link is proportional to ${d_{ij}^{-\gamma/2}}$. All simulations are obtained using $10^6$ samples for each fading channel. In these simulations the power allocation parameters are varied to obtain the best outage performance.

Figure \ref{fig:outvsnr} shows the outage probability for the composite DF scheme when the relay uses power based on the link-state as shown in Theorem \ref{thmoptimalpowerallocation}. Results confirm Corollary \ref{thmDA} as the proposed scheme achieves the full diversity order of $2$. These results also demonstrate that using partial relay power based on channel statistics may or may not degrade the outage performance compared to using full power, depending on the node-distance configuration.

In Fig. \ref{fig:ahmadfig7}, we consider a 2D network with fixed source location at (-5,0) and destination location at (5,0) while letting the relay move on the plane. Source and relay power are obtained from \eqref{snreq} such that the SNR is equal to 10dB. Further, we assume the relay only has long-term CSI for all links. This figure illustrates the regions in which the composite scheme has an outage probability below 2\% when the relay uses the minimum required power based on the link state. The results in Fig. \ref{fig:ahmadfig7} demonstrate that for a significant portion of the region in which the outage probability is below 2\%, the relay conserves power using the composite DF scheme with only long-term CSI. This result is significant because it establishes that the composite scheme conserves power at the relay while still satisfying the outage threshold.

\begin{figure}[t]
\begin{center}
\includegraphics[width=0.45\textwidth, height=56mm]{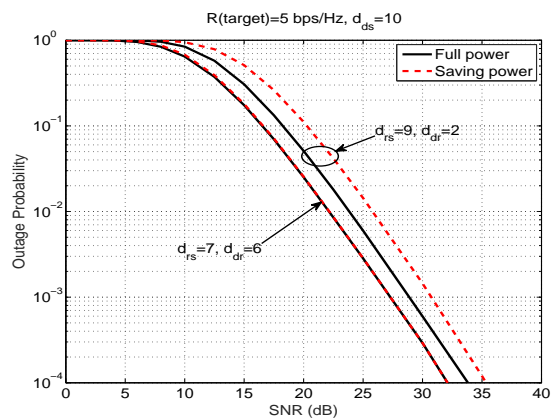}
    \caption{Outage probability for the composite DF scheme versus $\text{SNR}$ with partial or full relay power. Here the relay only has long-term CSI of all links.}\label{fig:outvsnr}
\end{center}
\end{figure}

\begin{figure}[t]
\begin{center}
\includegraphics[width=0.45\textwidth]{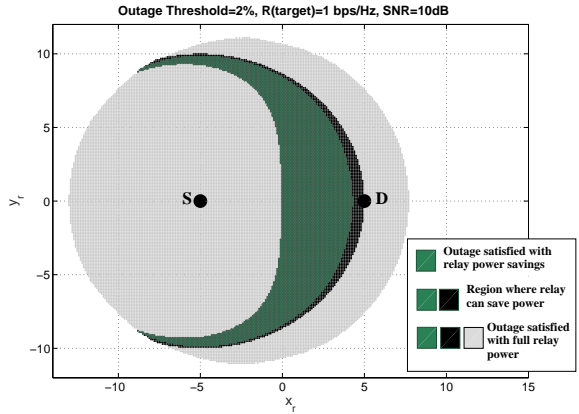}
    \caption{Regions where the proposed scheme achieves an outage probability below 2\% when the relay uses the minimum required power with long-term CSI.}
    \label{fig:ahmadfig7}
\end{center}
\end{figure}

\subsection{Relay Power Savings}\label{sec:numresultsPwr}
By simulating the expected relay power savings under both the perfect CSI and practical CSI assumptions, we verify the analysis in Theorem \ref{thmExpectedPwrSavings} and numerically evaluate the power savings in a plane.

\begin{figure}[t]
    \begin{center}
    \includegraphics[width=0.43\textwidth]{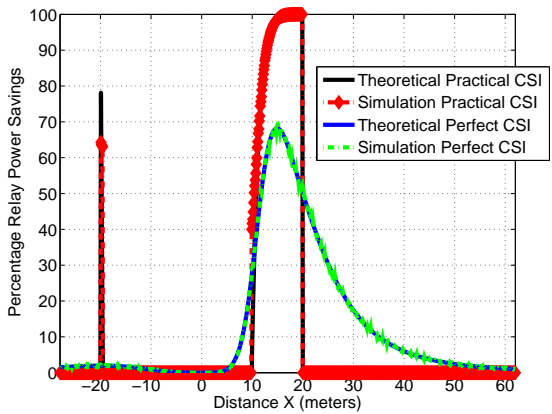}
    \caption{Comparison of relay power savings with practical and perfect CSI ($\gamma=3.6$, SNR$=5$dB). Note that this figure corresponds to the relay moving between the source and destination on line $y=0$ in Fig. \ref{fig:pwrSavings}.}
     \label{fig:compareCSI}
    \end{center}
\end{figure}
\begin{figure}[t]
   \begin{center}
        \subfigure[]{
		\includegraphics[width=0.42\textwidth]{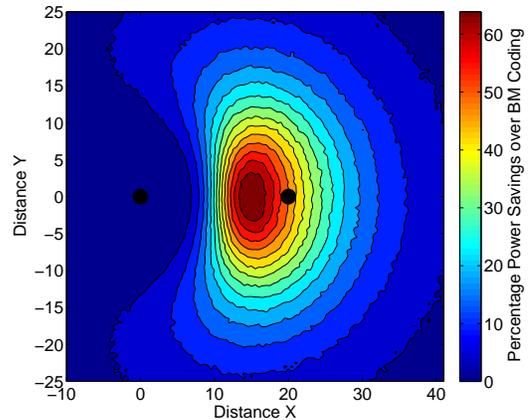}
		\label{fig:pwrperfectCSI}
    }
        \subfigure[]{
				\includegraphics[width=0.43\textwidth]{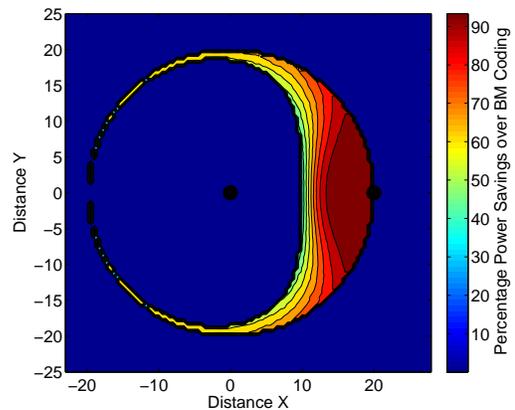}
				\label{fig:pwrpracCSI}
    }
    \caption{Percentage of relay power savings ($\gamma\!=\!3.6$; SNR$=5$dB) a) with perfect CSI b) with practical CSI}
    \label{fig:pwrSavings}
    \end{center}
\end{figure}

\subsubsection{Verification of Analytical Relay Power Savings}
First, we validate the expected relay power savings with perfect CSI and practical CSI in Theorem \ref{thmExpectedPwrSavings}. In Fig. \ref{fig:compareCSI}, the source is fixed at [0,0] and destination at [20,0] while the relay location varies along the X axis. For each fading channel in the perfect CSI case, the appropriate link-state regime is determined based on Theorem \ref{tablethm} and the relay power savings are computed as in (\ref{requiredRpowerR1}) for link-state regime $\mathcal{R}1$ (recall in $\mathcal{R}2$ the relay uses full power). In the practical CSI case, the link-state regime is determined based on the average links but the power consumed at the relay depends on the instantaneous receive SNR, $\gamma_s$, due to receiver channel estimation at the relay. It is evident in Fig. \ref{fig:compareCSI} that the simulation verifies Theorem \ref{thmExpectedPwrSavings} for both CSI assumptions. As the relay approaches either the destination or $-20$m, with practical CSI the relay conserves most of its power because the difference between the instantaneous $g_{rs}$ link amplitude and the average $\overline{g_{ds}}$ link amplitude is small.

Under both the perfect and practical CSI assumptions, the relay conserves power when it is between the source and destination. These savings occur at more relay locations for the perfect CSI case because with perfect CSI, link-state regime $\mathcal{R}1$ occurs with a non-zero probability in many inter-node distance configurations. However with practical CSI, relay power savings only occur when the node distances correspond directly to link-state regime $\mathcal{R}1$ via path loss only. Because of this, the practical CSI relay savings abruptly drops off while the perfect CSI relay savings changes more gradually as the relay location varies.

\subsubsection{Relay Power Savings in a Plane}
Next we simulate the relay power savings in two dimensions with both perfect and practical CSI in Fig. \ref{fig:pwrSavings}. In Fig. \ref{fig:pwrperfectCSI}, we assume the source and relay have perfect CSI and the power savings are computed for each relay location by averaging over all sample channels. In Fig. \ref{fig:pwrpracCSI}, we assume the source and relay only have practical CSI and the average channel proportional to the distance is used to both compute the link-state regime and power allocation parameters (with the exception of $\beta_{r,\text{prac}}$ from \eqref{optPwrAllocationRxSNR2} in link-state regime $\mathcal{R}1$). Thus the power savings are computed as in (\ref{requiredRpowerR1}) at each relay location.

A significant result from Fig. \ref{fig:pwrperfectCSI} in which the source and relay have perfect CSI is the vast proportion of space that the relay can conserve power and still achieve the maximum DF rate. As expected, the relay saves the most power (over $60\%$) in regions in which, without fading, the relay performs only independent coding. This is evident by comparing Fig. \ref{fig:pwrSavings} with Fig. \ref{fig:simulation}. When link-state regime $\mathcal{R}1$ is more probable, the relay is able to conserve more power.

We next compare the power savings at the relay under the two different CSI assumptions in Fig. \ref{fig:pwrSavings}. It is evident that relay power savings are possible in many more locations if perfect CSI is available. Specifically, if perfect CSI is available and the relay is within approximately 10 meters (half the source-to-destination distance) in any direction of the destination, relay power savings of over 25\% are attainable. Notably, even if the source-to-relay distance is larger than the source-to-destination distance, the relay can conserve power because of fading. Conversely, if practical CSI is available, the relay only conserves power if it is within 10 meters of the destination and is between the source and destination. Thus, in the practical CSI case, the relay only conserves power if the source-to-relay distance is less than the source-to-destination distance.

\begin{figure}[t]
   \begin{center}
    \includegraphics[width=0.5\textwidth]{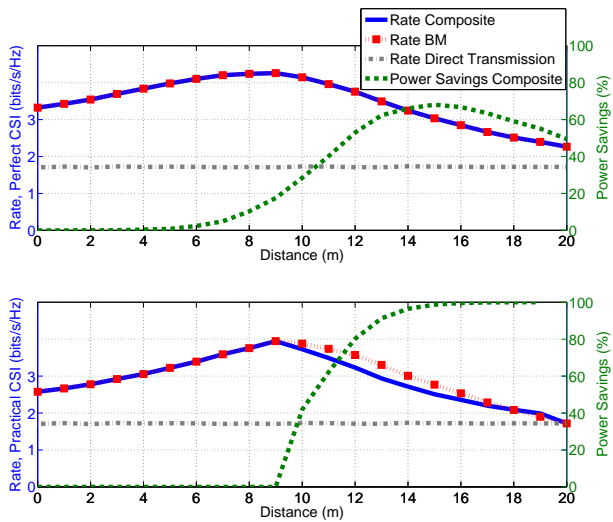} 
    \caption{Rate versus power savings comparison between block Markov DF relaying and composite DF relaying using (i) perfect CSI and (ii) practical CSI ($\gamma=3.6$, SNR$=5$dB)}\label{fig:rateVpower}      
    \end{center}
\end{figure}

\subsection{Power Savings and Rate Gain Trade-off for Relay Placement}\label{sec:tradeoff}
Here we illustrate the trade-off between rate gain and relay power savings for relay placement with perfect CSI or practical CSI. This trade-off is evident by comparing Fig. \ref{fig:rateGain3} with Fig. \ref{fig:pwrpracCSI} and in Fig. \ref{fig:rateVpower}. The most rate gain is obtained when the relay is closer to the source and the most power savings are realized when the relay is closer to the destination. Specifically if the relay is between the source and destination but closer to the source, employing the composite scheme with practical CSI results in up to 120\% rate gain over direct transmission in Fig. \ref{fig:rateGain3}, a significant gain. When the relay is closer to the destination, rate gain is reduced but there are greater relay power savings, as shown in Fig. \ref{fig:rateVpower}. As the relay approaches the destination, the average relay-to-destination link amplitude is very large due to the proximity of relay and destination; thus $\beta_{r,\text{prac}}$ approaches zero and the relay power savings approach 100\% in the limit.

On the other hand, with perfect CSI  shown in Fig. \ref{fig:pwrperfectCSI}, the link-state regime choice changes based on the instantaneous link states and as the relay approaches the destination, the relay is not always utilized because sometimes the source-to-relay link amplitude $g_{rs}$ will be less than that of the direct link, $g_{ds}$. Therefore, the power savings are not quite as large with perfect CSI as with practical CSI but perfect CSI achieves a slightly higher rate than that of practical CSI. However, it only requires forfeiting a small amount of rate gain to realize these substantial power savings when moving from perfect CSI to practical CSI.

In Fig. \ref{fig:rateVpower}, we compare the rate and power savings between the composite scheme and traditional block Markov DF relaying. As expected, with perfect CSI (top figure) the rate of the composite scheme is exactly the same as that of block-Markov DF relaying. With practical CSI (bottom figure), the rate of block Markov DF relaying is only slightly higher than that of the composite scheme in link-state regime $\mathcal{R}1$, and is identical in all other cases. Notably, however, the amount of power savings by the composite scheme is substantial under both CSI cases. A relay employing block Markov coding always uses full power, whereas a relay employing the composite scheme always saves power under  independent coding. Independent coding is also simple to implement practically because there is no requirement of source-relay phase coherency.

\section{Conclusion}\label{sec:conc}
In this paper, we analyze a composite DF scheme for the relay channel consisting of coherent block Markov coding and independent coding. By maximizing the achievable rate, we identify link-state regimes in which a particular transmission technique is optimal. We demonstrate that independent coding is optimal in one link-state regime which achieves the maximum DF rate but also results in power savings at the relay. The application of these link-state regimes for link adaptation in fading is analyzed under different CSI assumptions: perfect, long-term, and practical CSI. The expected relay power savings are computed in closed form for both perfect CSI and practical CSI. With any CSI model there is an implicit relay power savings versus rate trade-off for relay placement. The relay conserves the most power when it is closer to the destination but there is a more significant rate gain when the relay is closer to the source. We further derive the outage probability of the composite scheme with Rayleigh fading in closed form, demonstrating a full diversity order of $2$. These composite DF scheme results help in realizing efficient application of relaying in future cellular and wireless ad hoc networks.

\appendices
\section*{Appendix A: Proof of Lemma \ref{lemmaUsersFullPwr}, Theorem \ref{tablethm}, and Theorem \ref{thmoptimalpowerallocation} }\label{appendixA}
\textit{1. Proof of Lemma 1}
Suppose we fix $\beta_r$ and choose a power allocation such that $\alpha_s+\beta_s\!<\!P_s$. From the constraints in (\ref{optimization}) and rate expressions in (\ref{regionDF}), then we can always increase $\alpha_s$ and $\beta_s$ slightly while keeping $k_s\alpha_s$ the same. This increases both $J_1$ and $J_2$, which increases the rate. Therefore, to obtain the highest rate, $\alpha_s+\beta_s\!=\!P_s$ and the source always uses full power.

Next, let $\alpha_s\!>\!0$ so that the source performs block Markov coding. Then if $k_s\alpha_s+\beta_r\!<\!P_r$, $\alpha_s$ can be decreased slightly and $k_s$ increased such that both $J_1$ and $J_2$ increase, which means the rate also increases. Therefore, the relay uses full power ($k_s\alpha_s+\beta_r\!=\!P_r$) when $\alpha_s\!>\!0$.

\textit{2. Proof of Theorem \ref{tablethm} and Theorem \ref{thmoptimalpowerallocation}}

The complementary slackness conditions for (\ref{LagDT}) can be written as
\begin{align}
   \nabla\mathcal{L} &= \lambda_1^\star(J_1-R^\star) = \lambda_2^\star(J_2-R^\star) =\lambda_3^\star(P_s\!-\!\alpha_s^\star\!-\!\beta_s^\star) \nonumber\\
   &\ =
  \lambda_4^\star(P_r\!-\!k_s\alpha_s^\star\!-\!\beta_r^\star)=\lambda_{5}\alpha_s^*=\lambda_{6}\beta_r^*=0\nonumber
\end{align}
where all the primal and dual variables are non-negative.

\subsubsection{Case 1: $\lambda_1>0$, $\lambda_2=0$}
In this case, the derivative of the Lagrangian with respect to $\alpha_s$ reduces to
\begin{align}\label{simpderivDT}\
  \nabla_{\alpha_s}\mathcal{L} &=\lambda_3  + \lambda_4 k_s-\lambda_{5}=0.
\end{align}
Since the source transmits with full power (Lemma \ref{lemmaUsersFullPwr}), $\lambda_3>0$, then $\lambda_{5}$ must also be strictly positive. By complementary slackness, the minimum must occur when $\alpha_s\!=\!0$, implying that the source performs independent coding in this case.

Furthermore, we note that $J_1$ and $J_2$ are decreasing and increasing functions in $\alpha_s$ respectively, and both are positive. If the minimum occurs when $\alpha_s\!=\!0$, then this implies $\max\limits_{\alpha_s}J_1 \leq \min\limits_{\alpha_s}J_2$. Based on (\ref{regionDF}), this is equivalent to:
\begin{align}\label{channelcondDT1}
g_{rs}^2P_s&\le g_{ds}^2P_s+g_{dr}^2\beta_r.
\end{align}
Since $\beta_r\in[0,P_r]$, this case applies for all links that satisfy:
\begin{align}\label{channelcondDT2}
g_{ds}^2&<g_{rs}^2\le g_{ds}^2+\frac{P_r}{P_s}g_{dr}^2.
\end{align}
Rearranging (\ref{channelcondDT1}) gives a lower bound for $\beta_r$:
\begin{align}\label{beta3DT2}
\beta_r&\geq\frac{g_{rs}^2-g_{ds}^2}{g_{dr}^2}P_s.
\end{align}
$\beta_r$ is the minimum value that satisfies (\ref{beta3DT2}) as using a larger $\beta_r$ does not increase transmission rate. This minimum value of $\beta_r$ also saves relay power.

\subsubsection{Case 3: $\lambda_1>0$, $\lambda_2>0$}
For this case, we have $R=J_1=J_2$, which implies a crossing between $J_1$ and $J_2$. We consider crossings for $\alpha_s>0$, since $\alpha_s=0$ leads to the first case. By Lemma \ref{lemmaUsersFullPwr}, $k_s\alpha_s+\beta_r=P_r$ and $\lambda_4>0$. Moreover, $\alpha_s>0$ implies the source performs block Markov coding, which implies a nonzero $k_s$ and $\nabla_{k_s}\mathcal{L}=0$.

Combining the derivative of the Lagrangian with respect to $k_s$ with the derivative with respect to $\beta_r$ yields
\begin{align}
\lambda_2\frac{g_{ds}g_{dr}k_s^{-0.5}}
  {g_{ds}^2P_s+2g_{ds}g_{dr}\sqrt{k_s}\alpha_s
    +g_{dr}^2k_s\alpha_s+g_{dr}^2 \beta_r+1} - \lambda_{6}=0.\nonumber
\end{align}
Since $ \lambda_2>0$, then $\lambda_{6}>0$ as well. Therefore, by complementary slackness, the minimum must occur when $\beta_r=0$. This implies that the relay performs only block Markov coding, with no independent coding.

$\lambda_1>0$ and $\lambda_2>0$ implies a crossing between $J_1$ and $J_2$ since both are tight. Because $J_1$ is decreasing and $J_2$ is increasing in $\alpha_s$, it then must hold that $\max\limits_{\alpha_s}J_1 > \min\limits_{\alpha_s}J_2$. This condition is equivalent to:
\begin{align}\label{gr2constraintDT}
g_{rs}^2 &> g_{ds}^2+\frac{P_r}{P_s}g_{dr}^2.
\end{align}
In order to find the optimal value of $\alpha_s$, we must solve the equation $J_1=J_2$ (with $\beta_r=0$). Setting $J_1=J_2$ yields:
\begin{align}\label{alphasOpt}
g_{rs}^2\beta_s=g_{ds}^2P_s+2g_{ds}g_{dr}\sqrt{k_s}\alpha_s+g_{dr}^2k_s\alpha_s.
\end{align}
Since $\beta_r\!=\!0$, then $k_s\alpha_s\!=\!P_r$. Using this fact and the source's power constraint, (\ref{alphasOpt}) is a quadratic equation for $\sqrt{\alpha_s}$ and the squared positive root is the optimal value for $\alpha_s$ as in (\ref{optPwrAllocationC}).

\section*{Appendix B: Proof of Lemma \ref{outage formulation} and Theorem \ref{cor_out}}
\label{appendixB}
\subsection{Proof of Lemma \ref{outage formulation}}
In link-state regime $\mathcal{R}0$, direct transmission is optimal and outage can only occur at the destination; this outage is denoted as $\mathcal{P}_{dt}$. In link-state regime $\mathcal{R}1$, outage can only occur at the relay because $J_1 \leq J_2$ in this link-state regime. In link-state regime $\mathcal{R}2$, outage can occur separately at the relay or at the destination. Outage at the destination in $\mathcal{R}2$ is denoted as $\mathcal{P}_{dest}$. If we combine the outage at the relay in $\mathcal{R}1$ and in $\mathcal{R}2$, we obtain $\mathcal{P}_{relay}$.

For the outage formulation in Lemma \ref{outage formulation} to be correct the following identity should be true.
\begin{align}\label{outform1}
1-\bar{\mathcal{P}}_{out}&=\mathsf{Pr}\left[R\leq J_o| \mathcal{R}0\right] \mathsf{Pr}\left[ \mathcal{R}0\right]\nonumber\\
&\ \ \  +\mathsf{Pr}\left[R\leq \min\left(J_1,J_2\right)| \mathcal{R}1 \cup \mathcal{R}2  \right] \mathsf{Pr}\left[\mathcal{R}1 \cup \mathcal{R}2 \right]
\end{align}
where $\bar{\mathcal{P}}_{out}=\mathcal{P}_{dt}+\mathcal{P}_{relay}+\mathcal{P}_{dest}$ as in \eqref{outtc}. Note that all three link-state regimes are disjoint. As such, \eqref{outform1} can be simplified to the identity
\begin{align}\label{identity}
1=\mathsf{Pr}\left[ \mathcal{R}0 \right]+\mathsf{Pr}\left[ \mathcal{R}1 \right]+\mathsf{Pr}\left[ \mathcal{R}2 \right].
\end{align}

\subsection{Proof of Theorem \ref{cor_out}}
First, let $\tilde{g}_{dr}=\sqrt{\frac{P_r}{P_s}}g_{dr}$. Then, the outage probability for each case in \eqref{outtc} can be analyzed as follows.
\subsubsection{Outage Probability of Direct Transmission}
The analysis for the outage probability $\mathcal{P}_{dt}$ in (\ref{outevl}) is straightforward using the exponential distribution of
$g_{ds}^2$ and $g_{rs}^2$ as shown in \cite{myJDF}.
\subsubsection{Outage Probability at the Relay}
From (\ref{outtc}), the outage at the relay is given as follows.
\begin{align}\label{outagerelay1}
\mathcal{P}_{relay}&=\mathsf{Pr}\left[ \left(R> J_1\right) \cap \left(\mathcal{R}1\cup \mathcal{R}2\right)\right],\nonumber\\
&=\mathsf{Pr}\left[ R> \log\left(1+g_{rs}^2\beta_s\right),\;g_{rs}>g_{ds}\right],
\end{align}
where $\beta_s=P_s$ in $\mathcal{R}1$. Formula \eqref{outagerelay1} is similar to $\mathcal{P}_{dt}$ in (\ref{outtc}) except replacing
$\beta_s$ with $P_s$, $g_{rs}$ with $g_{ds}$ and $g_{ds}$ with $g_{rs}$. Hence, $\mathcal{P}_{relay}$ in (\ref{outevl}) is similar to $\mathcal{P}_{dt}$ except replacing $\beta_s$ with $P_s$, $\lambda_{rs}$ with $\lambda_{ds}$ and $\lambda_{ds}$ with $\lambda_{rs}$.

\subsubsection{Outage Probability at the Destination}
From (\ref{outtc}), the outage at the destination is given as follows.
\begin{align}\label{pbs}
\mathcal{P}_{dest}&=\mathsf{Pr}\big[\tilde{g}_{dr}^2P_s\!+\!g_{ds}^2P_s\!+\!2g_{ds}\tilde{g}_{dr}\sqrt{P_s\alpha_s}\!<\!2^R-1,\nonumber\\
&\ \ \ \ \ \ \ \ \  g_{rs}\!>\!\max\{\eta_1,\eta_3\}\big],\\
&=\mathsf{Pr}\big[g_{ds}\!<\!\beta_1,\;\tilde{g}_{dr}\!<\!\zeta_1,\;g_{rs}\!>\!\eta_1\big]\nonumber\\
&\ \ \ \ +\mathsf{Pr}\big[g_{ds}\!<\!\beta_1,\;\eta_2\!<\!\tilde{g}_{dr}\!<\!\zeta_1,\;g_{rs}\!>\!\eta_3\big],\nonumber\\
&=\mathsf{Pr}\big[g_{ds}\!<\!\beta_1,\; \tilde{g}_{dr}\!<\!\zeta_1,\; g_{rs}\!>\!\eta_1\big]\nonumber\\
&=\int_{0}^{\beta_1}\int_{0}^{\zeta_1}
\int_{\eta_1}^{\infty} f_1(g_{ds},\tilde{g}_{dr},g_{rs})dg_{rs}d\tilde{g}_{dr}dg_{ds}\nonumber
\end{align}
where
\begin{align}
&f_1(g_{ds},\tilde{g}_{dr},g_{rs})\!\!=\!\!8\lambda_{ds}\tilde{\lambda}_{dr}\lambda_{rs}g_{ds}\tilde{g}_{dr}g_{rs}
e^{-\left(\lambda_{ds}g_{ds}^2\!+\!
\tilde{\lambda}_{dr}\tilde{g}_{dr}^2\!+\!\lambda_{rs}g_{rs}^2\right)}\nonumber\\
&\eta_1=\sqrt{\frac{2^R\!-\!1}{\beta_s}},\ \eta_2=\sqrt{\frac{2^R\!-\!1}{\beta_s}-g_{ds}^2},\;\; \eta_3=\sqrt{g_{ds}^2\!+\!\frac{P_r}{P_s}g_{dr}^2},\nonumber\\
&\zeta_1=\frac{-g_{ds}\sqrt{\alpha_s}+\sqrt{g_{ds}^2(\alpha_s\!-\!P_s)\!+\!2^R\!-\!1}}{\sqrt{P_s}},\ \beta_1=\sqrt{\frac{2^R\!-\!1}{P_s}}.\nonumber
\end{align}
In (\ref{pbs}), $\mathsf{Pr}\big[g_{ds}<\beta_1,\;\eta_2<\tilde{g}_{dr}<\zeta_1,\;g_{rs}>\eta_3\big]=0$ because $\zeta_1<\eta_1$ for $0<g_{ds}<\beta_1$.
This can be shown as follows. Let $f(g_{12})=\zeta_1-\eta_1$, we have $f(0)<0$ and $f(\beta_1)=0$. Moreover, $f(g_{12})$ is a convex function
in $g_{12}$ since $\partial^2 f(g_{12})/\partial g_{12}^2>0$. Therefore, $\zeta_1<\eta_1$ for $0<g_{ds}<\beta_1$. Then, the analytical evaluations for the triple integral in \eqref{pbs} is straightforward as given in \eqref{outevl}.

\section*{Appendix C: Proof of Lemma \ref{theoremprLS} and Theorem \ref{thmExpectedPwrSavings}}

\subsection{Proof of Lemma \ref{theoremprLS}}
$\textsf{Pr} (\mathcal{R}0)$ follows by integrating over the joint PDF of two exponential random variables. For $\textsf{Pr} (\mathcal{R}1)$, define $W=g_{ds}^2+\frac{P_r}{P_s}g_{dr}^2$ and $X=g_{rs}^2$ for conciseness. From \cite{proofExp}, the PDF of W is
\begin{align}\label{wPDF}
f_{W}(w)=\frac{\lambda_{ds}\tilde{\lambda}_{dr}}{\tilde{\lambda}_{dr} - \lambda_{ds}}(e^{-\lambda_{ds} w}-e^{-\tilde{\lambda}_{dr} w}).
\end{align}
\begin{align}\label{probR2_3}
\text{Then }\textsf{Pr}\left(X<W\right) &=\mathbb{E}_{W}\left[ \textsf{Pr}\left( X<W | W=w \right) \right]\nonumber\\&=\int_0^\infty \left( \int_0^w f_X(x)dx\right) f_W(w) dw
\end{align}
Evaluating (\ref{probR2_3}) yields the probability of link-state regime $\mathcal{R}1$ in (\ref{prLinkStates}). The probability of link-state regime $\mathcal{R}2$ is obtained through the identity in \eqref{identity}.

\subsection{Proof of Theorem \ref{thmExpectedPwrSavings}}
This relay power savings is achieved when the relay performs independent coding for the source. The source-to-relay link $g_{rs}^2$ is bounded in these regions according to Theorem \ref{tablethm}. As such, we write the expected value of the relay power savings using perfect CSI as

\begin{align}\label{pwrSavings}
&\mathbb{E}[\mathcal{S}_{perf}]\!\!=\!\!\int_0^\infty \!\!\!\!\! \int_0^\infty \!\! \!\!\!\int_{g_{ds}^2}^{g_{ds}^2\!+\!\frac{P_r}{P_s}g_{dr}^2} \!\! \!\left(\!\!P_r\!-\!\frac{g_{rs}^2\!-\!g_{ds}^2}{g_{dr}^2}P_s\!\!\right)\!f(*)dg_{rs}^2 dg_{ds}^2 dg_{dr}^2\nonumber\\
&\text{where }\nonumber\\
&f(*)=\lambda_{rs}\lambda_{ds}\lambda_{dr}e^{\left[-\lambda_{rs}g_{rs}^2-\lambda_{ds}g_{ds}^2-\lambda_{dr}g_{dr}^2\right]}.
\end{align} 
Note that the expected value of relay power savings with perfect CSI is written in integral form in \eqref{pwrSavings} because link-state regime $\mathcal{R}1$ can occur in any node-distance configuration with some probability as demonstrated by Lemma \ref{theoremprLS}. Evaluating (\ref{pwrSavings}) yields (\ref{pwrSavings2}). 

With practical CSI in $\mathcal{R}1$, the relay uses instantaneous receive SNR $\gamma_s$ to compute $\beta_r$ (the relay does not have instantaneous CSI for $\gamma_o$ or $\gamma_d$ and hence uses the average channel statistics to calculate these). As such, the expected relay power savings with practical CSI is

\begin{align}
\mathbb{E}_{\gamma_s}[\mathcal{S}_{prac}] =
\left\{
	\begin{array}{ll}
	P_r\left(1-\frac{\gamma_s-\bar{\gamma}_o}{\bar{\gamma}_d-\bar{\gamma}_o}\right)  & \mbox{if } \bar{\gamma}_o < \bar{\gamma}_s \leq \bar{\gamma}_d \\
		0 & \mbox{else } 
	\end{array}
\right.
\end{align}
Evaluating this with Rayleigh fading yields (\ref{pwrSavingsltCSI_2}).
\section*{Acknowledgment}
This work has been supported in part by the Office of Naval Research (ONR, Grant N00014-14-1-0645) and National Science Foundation Graduate Research Fellowship Program (NSF, Grant No. DGE-1325256). Any opinions, findings, and conclusions or recommendations expressed in this material are those of the authors and do not necessarily reflect the views of the Office of Naval Research or NSF.

\bibliographystyle{IEEEtran}
\bibliography{references}
\end{document}